\documentclass[11pt, letterpaper]{article}
\usepackage[utf8]{inputenc}
\usepackage{fullpage,times}
\usepackage[T1]{fontenc}

\usepackage{caption}
\usepackage[linesnumbered,algoruled,boxed,lined,algo2e]{algorithm2e}
\usepackage{algorithm}

\RestyleAlgo{ruled}
\SetKwInOut{Parameter}{Parameters}
\SetKwComment{Comment}{/* }{ */}

\usepackage[T1]{fontenc}
\usepackage{amsthm,amsmath,amssymb,amsfonts,authblk,hyperref,setspace}
\usepackage{graphicx}
\usepackage{caption}
\usepackage{subcaption}
\usepackage{xcolor}
\usepackage{authblk}
\usepackage{enumitem}
\usepackage{stfloats}
\usepackage{url}
\usepackage{float}
\usepackage{tikz}

\usepackage{calc}
\usepackage{tikz}
\usepackage{hyperref}
\usepackage{graphicx}
\usepackage{tabularx}

\newtheorem{theorem}{Theorem}[section]

\newtheorem{lemma}[theorem]{Lemma}

\newtheorem{proposition}[theorem]{Proposition}

\newtheorem{corollary}[theorem]{Corollary}

\newenvironment{proofof}[1]{\emph{Proof of #1.  }}{\hfill$\Box$}

\DeclareMathOperator*{\ED}{ED}
\DeclareMathOperator*{\Ham}{Ham}

\newcommand\eps{\varepsilon}

\newcommand{\OO}{{\widetilde{O}}}
\newcommand{\cc}{{\mathtt{c}}}
\newcommand{\rr}{{\mathtt{r}}}

\newcommand{\N}{{\mathbb{N}}}
\newcommand{\FL}{F_{\mathrm{CVL}}}
\newcommand{\Process}{{\mathrm{Process}}}
\newcommand{\Split}{{\mathrm{Split}}}
\newcommand{\Compress}{{\mathrm{Compress}}}

\newcommand{\Grammar}{{\mathrm{Grammar}}}

\newcommand{\FKR}{{F_\mathrm{KR}}}
\newcommand{\Enc}{{\mathrm{Enc}}}
\newcommand{\Bin}{{\mathrm{Bin}}}
\newcommand{\MI}{{\mathrm{MIS}}}
\newcommand{\SKED}{\mathrm{sk}^{\mathrm{ED}}}
\newcommand{\SKHAM}{\mathrm{sk}^{\mathrm{Ham}}}
\newcommand{\SKR}{\mathrm{sk}^{\mathrm{Rolling}}}
\newcommand{\Append}{{\mathrm{Append}}}
\newcommand{\Remove}{{\mathrm{Remove}}}
\newcommand{\Compare}{{\mathrm{Compare}}}

\newcommand{\UpdateAG}{{\mathrm{UpdateActiveGrammars}}}
\newcommand{\PartialDecompress}{{\mathrm{PartiallyDecompress}}}
\newcommand{\Recompress}{{\mathrm{Recompress}}}
\newcommand{\RecompressFirstBlock}{{\mathrm{RecompressFirstBlock}}}

\newcommand{\FindCompressedPrefix}{{\mathrm{FindCompressedPrefix}}}
\newcommand{\SplittingDepth}{{\mathrm{SplittingDepth}}}
\newcommand{\DecompressSymbol}{{\mathrm{DecompressSymbol}}}
\newcommand{\DecompressString}{{\mathrm{DecompressString}}}
\newcommand{\DecompressSymbolLength}{{\mathrm{DecompressSymbolLength}}}
\newcommand{\CompressWithGrammar}{{\mathrm{CompressWithGrammar}}}
\newcommand{\CrossOverBlock}{{\mathrm{CrossOverBlock}}}

\newcommand{\acc}[1]{{\em \hfil // #1 }}

\newcommand{\poly}{{\mathrm{poly}}}

\newcommand{\Dict}{{\mathrm{Dict}}}
\newcommand{\eval}{{\mathrm{eval}}}

\title{Locally consistent decomposition of strings with applications to edit distance sketching}
\author[1]{Sudatta Bhattacharya\thanks{Email: sudatta@iuuk.mff.cuni.cz. Partially supported by the Grant Agency of the Czech Republic under the grant agreement no. 19-27871X.}}
\author[1]{Michal Kouck{\'{y}}\thanks{Email: koucky@iuuk.mff.cuni.cz. Partially supported by the Grant Agency of the Czech Republic under the grant agreement no. 19-27871X. This project has received funding from the European Union’s Horizon 2020 research and innovation programme under the Marie Skłodowska-Curie grant agreement No. 823748 (H2020-MSCA-RISE project CoSP).}}
\affil[1]{Computer Science Institute of Charles University,
Malostransk{\'e}  n{\'a}m\v{e}st\'{\i} 25,
118 00 Praha 1, Czech Republic}
\date{}

\begin{document}

\maketitle

\begin{abstract}
    In this paper we provide a new locally consistent decomposition of strings. Each string $x$ is decomposed into blocks
    that can be described by grammars of size $\OO(k)$ (using some amount of randomness). 
    If we take two strings $x$ and $y$ of edit distance at most $k$
    then their block decomposition uses the same number of grammars and the $i$-th grammar of $x$ is the same as the $i$-th grammar of $y$ except for at most $k$ indexes $i$. 
    The edit distance of $x$ and $y$ equals to the sum of edit distances of pairs of blocks where $x$ and $y$ differ.
    Our decomposition can be used to design a sketch of size $\OO(k^2)$ for edit distance, 
    and also a rolling sketch for edit distance of size $\OO(k^2)$. 
    The rolling sketch allows to update the sketched string by appending a symbol or removing a symbol from the beginning of the string. 
\end{abstract}

\section{Introduction}

Edit distance is a measure of similarity of two strings. 
It measures how many symbols one has to insert, delete or substitute in a string $x$ to get a string $y$. 
The measure has many applications from text processing to bioinformatics. 
The edit distance $\ED(x,y)$ of two strings $x$ and $y$  can be computed in time $O(n^2)$ by a classic dynamic programming algorithm~\cite{WF74}.
Save for poly-log improvements in the running time~\cite{MP80,G16}, 
the best known running time for edit distance computation is $O(n+k^2)$~\cite{LMS98}, where $k=\ED(x,y)$.
Assuming Strong Exponential Time Hypothesis (SETH) this running time cannot be substantially improved~\cite{BI15}.
The conditional lower bound does not exclude some approximation algorithms, though, and there was a recent progress on computing edit distance in almost-linear time to within some constant factor approximation~\cite{CDGKS18,KS20,BR20,AN20}.

Another problem for edit distance that saw a major progress in recent years is sketching. 
In sketching we want to map a string $x$ to a short sketch $\SKED_{n,k}(x)$ 
so that from sketches $\SKED_{n,k}(x)$  and $\SKED_{n,k}(y)$ of two strings $x$ and $y$ we can compute their edit distance, either exactly or approximately. Apriori it is not even obvious that short sketches for edit distance exist. In a surprising construction, 
Belazzougui and Zhang~\cite{belazzougui_zhang} gave an exact edit distance sketch of size $O(k^{8}\log^{5}n)$ bits.
The sketch size was then improved to $O(k^{3}\log^{2}(\frac{n}{\delta})\log{n})$ bits by Jin, Nelson and Wu~\cite{nelson_edit_sketch}, where the $\ED(x,y)$ was computed exactly from the sketches with probability at least $1-\delta$, if $\ED(x,y)\le k$.
The current best sketch is of size $O(k^{2}\log^{3}n)$ bits and was given by Kociumaka, Porat and Starikovskaya~\cite{editsketchfocs2021}.
\cite{nelson_edit_sketch} gives a lower bound $\Omega(k)$ on the size of a sketch for exact edit distance.

The major problem in edit distance computation as well as in sketching is how to align the matching parts of two strings $x$ and $y$.
Finding an optimal alignment of two strings is the crux in the computation of edit distance and its sketching. 
In sketching finding a good alignment is even more challenging as we do not have both strings in our hands simultaneously to look for the matching.
To the best of our knowledge, to resolve this issue all edit distance sketches use {\em CGK random walk} on strings~\cite{CGK16}
which allows to embed the edit distance metrics into Hamming distance metrics with distortion $O(k)$. 
The walk implicitly fixes some reasonably good matching between the two strings.
Going from the CGK random walk to a sketch is non-trivial undertaking and all three sketch results rely on sophisticated machinery to achieve it.

In this paper we provide a new technique to align two strings $x$ and $y$ in oblivious manner.
In nutshell, we provide a decomposition procedure that breaks $x$ and $y$ into the same number of ``short'' blocks so that at most $k$ pairs
of blocks in the decomposition of $x$ and $y$ differ, and all other pairs of blocks are matching in an optimal alignment.
So the edit distance of $x$ and $y$ is the sum of edit distances of the differing blocks.
To be more specific our blocks are not short in their length 
but they are short in the sense that each of them can be described by a  context-free grammar of size $\OO(k)$.
Our decomposition algorithm constructs the grammars.
Our decomposition is based on {\em locally consistent parsing} of strings a technique similar to the one used in~\cite{lcp-application94, BES06, Jow, lcp_application2020} and hash based partitioning similar to \cite{ZZ-19}. 
Our main technical result is:

\begin{theorem}[String decomposition]\label{t-main1}
    There is an algorithm running in time $\OO(|x|)$ that for each string $x$ of length at most $n$ produces grammars $G^x_1,\dots,G^x_s$ such that with probability at least $1-O(1/n)$, $x=\eval(G^x_1)\cdots \eval(G^x_s)$ and each of the grammars is of size $\OO(k)$. (The $\OO(\cdot)$ notation hides factors that are poly-logarithmic in $n$.) Furthermore, for any two strings $x$ and $y$ of edit distance at most $k$ with grammars $G^x_1,\dots,G^x_s$ and  $G^y_1,\dots,G^y_{s'}$, resp., that are produced by the algorithm using the same randomness, the following is true simultaneously with probability at least $4/5$:
    \begin{enumerate}
        \item $s=s'$,
        \item $G^x_i=G^y_i$, for all $i\in \{1,\dots,s\}$ except for at most $k$ indices $i$, and
        \item $\ED(x,y)=\sum_i \ED(\eval(G^x_i),\eval(G^y_i))$.
    \end{enumerate}
\end{theorem}

Here, for a grammar $G$, $\eval(G)$ denotes its evaluation. 
Our decomposition can be used immediately to give an embedding of edit distance into Hamming distance with distortion $O(k)$.
It also readily yields a sketch for exact edit distance of size $\OO(k^2)$:

\begin{theorem}[Sketch for edit distance]\label{t-main2}
   There is a randomized sketching algorithm $\SKED_{n,k}$ that on an input string $x$ of length at most $n$ produces a sketch $\SKED_{n,k}(x)$ 
   of size $\OO(k^2)$ in time $\OO(nk)$, and a comparison algorithm running in time $\OO(k^2)$ such that given two sketches $\SKED_{n,k}(x)$ and $\SKED_{n,k}(y)$ for two strings $x$ and $y$ of length at most $n$ obtained using the same randomness of the sketching algorithm outputs with probability at least $1-1/n$ (over the randomness of the sketching and comparison algorithms) the edit distance of $x$ and $y$ if it is less than $k$ and $\infty$ otherwise.
\end{theorem}
 
Furthermore, we can also provide a {\em rolling sketch}, a sketch in which we can update the stored string by appending a symbol or
removing its first symbol. 

\begin{theorem}[Rolling sketch for edit distance]\label{t-main3}
   There are algorithms $\Append(sk_x,a)$, $\Remove(sk_{ax},a)$, and $\Compare(sk_x, sk_y)$ such that for integer parameters $k\le m$:
   \begin{enumerate}
       \item Given a sketch $sk_x$ representing a string $x$ and a symbol $a$, $\Append(sk_x,a)$ outputs a sketch $sk_{xa}$ for the string $xa$ in time $\OO(k^2)$.
       \item Given a sketch $sk_{ax}$ representing a string $ax$ for a symbol $a$, $\Remove(sk_{ax},a)$ outputs a sketch $sk_{x}$ for the string $x$ in time $\OO(k^2)$.
       \item Given two sketches $sk_{x}$ and $sk_y$ representing strings $x$ and $y$ obtained from the same random sketch for empty string using two sequences of at most $m$ operations $\Append$ and $\Remove$, $\Compare(sk_x, sk_y)$ calculates the edit distance of $x$ and $y$
       if it is less than $k$, and outputs $\infty$ otherwise. The algorithm $\Compare(sk_x, sk_y)$ runs in time $\OO(k^2)$.
   \end{enumerate}
   All the sketches are of size $\OO(k^2)$.
   The probability that any of the algorithms fails or produces incorrect output is at most $1/m$ over the initial randomness of the sketch for empty string and internal randomness of the algorithms. 
\end{theorem}

We remark that we did not attempt to optimize the running time of either of our algorithms, or poly-log factors in the sketch sizes, 
and we believe that both parameters can be readily improved by usual amortization techniques of processing symbols in batches of size $\OO(k)$.
We believe that the update time in the last theorem can be improved to $\OO(1)$ by buffering $\OO(k)$ symbols
that shall be inserted or removed without affecting the other parameters of the algorithm.

Another distinguishing feature of our decomposition procedure compared to the technique of CGK random walks is its parallelizability.
CGK random walk seems inherently sequential whereas our decomposition procedure can be easily parallelized. 
We believe that our decomposition will allow for further applications beyond our simple sketches.

\subsection{Related work}

The problem of embedding edit distance to other distance measures, like Hamming distance, $\ell_1$, etc. has been studied extensively. In \cite{CGK16}, the authors have given a randomized embedding from edit distance to Hamming distance, where any string $x \in \{0,1\}^n$ can be mapped to a string $f(x)\in \{0,1\}^{3n}$, given a random string $r\in\{0,1\}^{\log^{2}n}$, such that, $\ED(x,y)/2 \le \Ham(f(x),f(y)) \le O(\ED(x,y)^2)$ with probability at least $2/3$.
% (the random string $r$ is fixed for a particular embedding).
Batu, Ergun and Sahinalp~\cite{BES06} have introduced a dimensionality reduction technique, where any string $x$ of length $n$ can be mapped to a string $f(x)$ of length at most ${n}/{r}$, for any parameter $r$, with a distortion of $\OO(r)$. 
They used the locally consistent parsing technique for their embedding.
Ostrovsky and Rabani \cite{hamming_to_l1_rabani_2007} gave an embedding from edit distance to $\ell_1$ distance with a distortion of $O(\sqrt{\log n \log \log n})$. 
Jowhari \cite{Jow} also gave a randomized embedding from edit distance to $\ell_1$ distance with a distortion of $O(\log n \log^{*}n)$. He used the embedding given by Cormode and Muthukrishnan \cite{CM02} who showed that any string $x$ of length $n$ can be mapped to a vector $f(x)$ of length $m=O(2^{n \log n})$, such that for any pair of strings $x,y$ of length $n$ each, $\ED(x,y)/2 \le \lVert f(x)-f(y) \rVert_{\ell_1} \le O(\log n \log^{*}n) \cdot \ED(x,y)$. Since the size of the vector was too large, \cite{Jow} used random hashing to get his final embedding. 
%\cite{ADGIR03,KN05}

\subsection{Our techniques}

We first provide the intuition for our technique. We would like to break a string $x$ into small blocks {\em obliviously} 
so that when a string $y$ is broken by the same procedure, the difference between $x$ and $y$ caused by the edit operations
is confined within the corresponding blocks of $x$ and $y$, and the overall decomposition is not affected by them.
For random binary strings $x$ and $y$ this could be done fairly easily: 
look on all the (overlapping) windows of $\log n$ consecutive bits in each of the strings and for each window decide at random whether to make a break at that window or not. 
To make it consistent between $x$ and $y$ use some random hash function $H:\{0,1\}^{\log n} \rightarrow \{0,\dots,D-1\}$
so that if the hash function evaluates to $0$ on a given window then start a next block of the decomposition.
If we chose $D$ suitably, say $D \ge 10 k \log n$, then we are unlikely to start a new block in any window which is affected by the the at most $k$ edit operations on $x$ and $y$. In that case we obtain the desired decomposition. Hence, decomposing random strings $x$ and $y$ is easy.

The issue is what to do with non-random strings. 
Consider for example strings $x$ and $y$ that are very sparse, so they contain $\sqrt{n}$ ones sprinkled within the vast ocean of zeros. 
The hash function $H$ will see mostly windows of 0's and occasionally a window of the form $0^i10^{\log(n) - i -1}$.
The decomposition will have no effect on such strings despite the fact that the string might contain $\Omega(\sqrt{n})$ bits of entropy.

However, we can compress such sparse strings: replace stretches of zeros by some binary encoded information about their length,
and try to break the strings again. 
Still, this will fail if in our example the stretches of zeros are replaced by stretches of some repeated pattern such as $(01)^*$.
So we need slightly more general compression which will compress any $\log n$ bits into $\log(n)/2$ bits. 
By repeating the sequence of steps: split and compress, we will eventually get the desired decomposition of each string.

Our actual algorithm mimics the above intuition. 
It is technically easier to work with a larger alphabet, so we extend the input alphabet $\Sigma$ by adding special compression symbols into the work alphabet $\Gamma$. 
(Without loss of generalization we can assume that $\Sigma$ is of size $O(n^3)$ otherwise we can hash each symbol of our input strings using some perfect hash function into an alphabet of size $O(n^3)$ without affecting the edit distance of a given pair of strings.)
To split a string we will use a random hash function $H:\Gamma^2 \rightarrow \{0,1\}$ from a suitable hash family that we call {\em $(D,O(\log n))$-iterated pair-wise independent family}, for $D=\Theta(k \log n)$.\footnote{In earlier version of this paper
we used $H$ to be $\OO(k)$-wise independent hash function $H:\Gamma^2 \rightarrow \{0,\dots, D-1\}$. 
In the current version we replace it by a choice from a smaller hash family which is computationally more efficient.}
If the hash function is zero on a pair of consecutive symbols in a string, we start a new block of the decomposition on the first symbol in the pair, and this happens with probability roughly $1/D$ for our choice of $H$.

Then in each resulting block we replace stretches of repeated symbols by a special compression symbol from $\Gamma$ representing the block,
and we use a pair-wise independent hash function $C:\Gamma^2 \rightarrow (\Gamma \setminus \Sigma)$ to compress non-overlapping pairs of symbols
into one symbol. This latter step requires some care as we have to make sure that we select non-overlapping pairs in the same way in $x$ and $y$.
For the selection of non-overlapping pairs we use the locally consistent coloring of Cole and Vishkin~\cite{cole1986deterministic,linial1987distributive,linial1992locality} where the selection of pairs depends only on the context of $O(\log^* n)$ symbols.
The compression reduces the size of each block by a factor of $2/3$.
We repeat the compress and split process  for $O(\log n)$ iterations until each compressed block of $x$ is of size at most 2. 
{\em Decompression} of each block then gives us the desired decomposition of $x$. (See Fig.~\ref{fig:hierarchy} for an illustration.)

It is natural and convenient to represent each of the blocks by a context-free grammar which corresponds to the compression process.
We can argue that the grammars will be of size $O(D \log n)$ with high probability. 
So we can represent each string by a sequence of small grammars so that if $x$ and $y$ are at edit distance at most $k$ 
then at most $k$ pairs of their grammars will differ, and the sum of the edit distances of differing pairs is the edit distance of $x$ and $y$.
Note, that edit distance of two strings represented by context-free grammars can be computed efficiently~\cite{ED_compressed_string_Soda22}.
These are the main ideas behind our decomposition algorithm, and we provide more details in Section~\ref{s-decomposition}

\smallskip
Building a sketch from the string decomposition is straightforward: 
We encode each grammar in binary using fixed number of bits, and we use off-the-shelf sketch for Hamming distance to sketch the sequence of grammars.
As the Hamming distance sketch does not recover identical bits but only the mismatched bits we make sure that if two grammars differ then 
their binary encoding differ in every bit. 
Over binary alphabet this might be impossible but over large alphabets one could use error-correcting codes to achieve the desired effect of recovering the differing grammars; for simplicity we use the Karp-Rabin fingerprint of the whole grammar to encode the binary 0 and 1 distinctly.
See Section~\ref{s-binencoding} for the details of our encoding and Section~\ref{s-edsketch} for details of the sketch for edit distance.

\smallskip
To design a rolling sketch for edit distance where we can extend the represented string by a new symbol or repeatedly remove the first symbol
of the represented string 
we will employ our decomposition technique together with the rolling sketch for Hamming distance of Clifford, Kociumaka, and Porat~\cite{rollinghashSODA2019}.
We will argue that appending a new symbol to a string affects only some fixed number of grammars in the decomposition of a string.
There is a certain threshold $T$ so that except for the last $T$ grammars the decomposition of a string stays the same regardless of how many 
other symbols are appended. 
Hence, we will keep a buffer of at most $T$ {\em active} grammars corresponding to the recently added symbols, 
and upon addition of a new symbol we will only update those grammars. 
We are guaranteed that the grammars before this threshold will stay the same forever, so we can {\em commit} them into the rolling Hamming sketch
(in the form of their binary encoding.)
Similarly, we will keep a buffer of up-to $T$ {\em active} grammars that capture the symbols that were deleted from the sketch most recently. 
Once they become ``mature'' enough we can commit them by removing their binary encoding from the rolling Hamming sketch.
(See Fig.~\ref{fig:rolling_sketch} for an illustration.)
This allows to maintain a rolling sketch for edit distance.

Evaluation of an edit distance query on two rolling sketches will use their Hamming sketch to recover differing committed grammars. 
Together with the active grammars of inserted and deleted symbols this provides enough information for evaluating the edit distance query.
Technical details are explained in Section~\ref{s-rolling}.
In Section~\ref{s-param} we give a table of parameters used throughout the paper.

\section{Notations and preliminaries}
% We use $[n]$ to represent the set of integers $[n]=\{1,2,\dots,n\}$.
For any string $x = x_1x_2x_2\dots x_n$ and integers $p,q$, $x[p]$ denotes $x_p$, $x[p,q]$ represents substring $x' = x_p\dots x_q$ of $x$, and $x[p,q)=x[p,q-1]$.
If $q<p$, then $x[p,q]$ is the empty string $\eps$. $x[p,\dots]$ represents $x[p,|x|]$, where $|x|$ is the length of $x$.
"$\cdot$"-operator is used to denote concatenation, e.g $x\cdot y$ is the concatenation of two strings $x$ and $y$. $\Dict(x) = \{x[i,i+1], i \in [n-1]\}$, is the dictionary of string $x$, which stores all pairs of consecutive symbols that appear in $x$.
For strings $x$ and $y$, $\ED(x,y)$ is the minimum number of modifications ({\em edit operations}) required to change $x$ into $y$, where a single modification can be adding a character, deleting a character or substituting a character in $x$.
All logarithms are based-2 unless stated otherwise.
For integers $p>q$, $\sum_{i=p}^{q} a_i=0$ by definition regardless of $a_i$'s.

\subsection{Grammars}\label{s-grammars}

Let  $\Sigma \subseteq \Gamma$ be two alphabets and $\# \not\in \Gamma$. A {\em grammar} $G$ is a set of {\em rules} of the type $c \rightarrow ab$
or $c \rightarrow a^r$, where $c \in (\Gamma \cup \{\#\}) \setminus \Sigma$, $a,b \in \Gamma$ and $r\in \N$. 
$c$ is the {\em left hand side} of the rule, and $ab$ or $a^r$ is the {\em right hand side} of the rule. 
$\#$ is the starting symbol. The size $|G|$ of the grammar is the number of rules in $G$.  
We only consider grammars where each $a \in \Gamma \cup \{\#\}$ appears on the left hand side of at most one rule of $G$, we call such grammars {\em deterministic}. 
(We assume that rules of the form $c \rightarrow a^r$ are stored in implicit (compressed) form.) 
The $\eval(G)$ is the string from $\Sigma^*$ obtained from $\#$ by iterative rewriting of the intermediate results by the rules from $G$. 
If the rewriting process never stops or stops with a string not from $\Sigma^*$, $\eval(G)$  is undefined.
Observe, that we can replace each rule of the type $c \rightarrow a^r$ by a collection of at most $2\lceil \log r \rceil$ new rules of the other type using some auxiliary symbols. 
Hence, for each grammar $G$ there is another grammar $G'$ using only the first type of the rules such that $\eval(G)=\eval(G')$ and $|G'| \le |G| \cdot 2 \lceil  \log |\eval(G)| \rceil$. 
Using a depth-first traversal of a deterministic grammar $G$ we can calculate its {\em evaluation size $|\eval(G)|$} in time $O(|G|)$. 
Given a deterministic grammar $G$ and an integer $m$ less or equal to its evaluation size, we can construct in time $O(|G|)$
another grammar $G'$ of size $O(|G|)$ such that $\eval(G')=\eval(G)[m,\dots]$. $G'$ will use some new auxiliary symbols.
Given a deterministic grammar $G$, using a depth-first traversal on symbols reachable from the starting symbol $\#$
we can identify in time $O(|G|)$ the smallest sub-grammar $G'\subseteq G$ with the same evaluation.

We will use the following observation of Ganesh, Kociumaka, Lincoln and Saha \cite{ED_compressed_string_Soda22}:

\begin{proposition}[\cite{ED_compressed_string_Soda22}]\label{p-edgrammar}
There is an algorithm that on input of two grammars $G_x$ and $G_y$ of size at most $m$ computes the edit distance $k$ of $\eval(G_x)$ and $\eval(G_y)$ in time $O( (m + k^2) \cdot \poly(\log m+n))$, where $n=|\eval(G_x)|+|\eval(G_y)|$.
\end{proposition}

\subsection{Rolling Hamming distance sketch}\label{s-hammin}

For two strings $x$ and $y$ of the same length, we define their {\em mismatch information} 
$\MI(x,y)=\{(i,x[i],y[i]);$ $i \in \{1,\dots,|x|\} \textit{ and }x[i]\neq y[i]\}$. 
The Hamming distance of $x$ and $y$ is $\Ham(x,y)=|\MI(x,y)|$.

There exist various sketches for Hamming distance, which allow to compute Hamming distance with low error probability \cite{hamming_sketch2,hamming_sketch3}. Moreover, \cite{porat2007,rollinghashSODA2019} also allow to retrieve the mismatch information. 
For our purposes we will use the sketch given by Clifford, Kociumaka, and Porat~\cite{rollinghashSODA2019}.

Let $k\le n$ be integers and $p \ge n^3$ be a prime. 
\cite{rollinghashSODA2019}  give a randomized sketch for Hamming distance 
$\SKHAM_{n,k,p} : \{1,\dots,p-1\}^* \rightarrow \{0,\dots,p-1\}^{k+4}$ computable in time $\OO(n)$ with the following properties.\footnote{Clifford, Kociumaka and Porat have the sketch size only $k+3$ elements but we include as an extra item the randomness of the sketch, which is a single element from $\{0,\dots,p-1\}$
used to compute Karp-Rabin fingerprint.}

\begin{proposition}[\cite{rollinghashSODA2019}]\label{p-skham}
    There is a randomized algorithm working in time $O(k \log^3 p)$ that given sketches $\SKHAM_{n,k,p}(x)$
    and $\SKHAM_{n,k,p}(y)$ of two strings $x$ and $y$ of length $\ell \le n$ constructed using the same randomness decides whether $\Ham(x,y)\le k$, and if so returns $\MI(x,y)$,
    with probability of error at most $1/n$ over the randomness of the sketches and the internal randomness of the algorithm.
\end{proposition}

They also construct the following update procedures for their sketch. We will use them to construct a rolling sketch for edit distance.

\begin{proposition}[Lemma 2.3 of \cite{rollinghashSODA2019}]
For $x\in \{1,\dots, p\}^*$ of length less than $n$ and $a \in \{1,\dots, p\}$, in time $O(k \log p)$ we can compute:
\begin{enumerate}
    \item $\SKHAM_{n,k,p}(xa)$ and $\SKHAM_{n,k,p}(ax)$, given $\SKHAM_{n,k,p}(x)$ and $a$.
    \item $\SKHAM_{n,k,p}(x)$ given $\SKHAM_{n,k,p}(xa)$ or $\SKHAM_{n,k,p}(ax)$, and $a$.
\end{enumerate}
\end{proposition}

Corollary 2.5 of \cite{rollinghashSODA2019} states that appending a character to a sketch of $x$ can be done even faster namely in amortized time $O(\log p)$.

\subsection{Locally consistent coloring}\label{s-locally_consistent_parsing}

The following color reduction procedure allows for locally consistent parsing of strings. The technique was originally proposed by Cole and Vishkin \cite{cole1986deterministic} and further studied by Linial \cite{linial1987distributive,linial1992locality}.

\begin{proposition}[\cite{cole1986deterministic,linial1987distributive,linial1992locality}]\label{thm:color_reduction}
    There exists a function $\FL:\Gamma^* \rightarrow \{1,2,3\}^*$ with the following properties. Let $R= \log^* |\Gamma| + 20$. For each string $x \in \Gamma^*$ in which no two consecutive symbols are the same:
    \begin{enumerate}
        \item $|\FL(x)|=|x|$ and $\FL(x)$ can be computed in time $O(R\cdot |x|)$.
        \item For $i\in \{1,\dots, |x|\}$, the $i$-th symbol of $\FL(x)$ is a function of symbols of $x$ only in positions $\{i-R,i-R+1\dots,i+R\}$.
        \item No two consecutive symbols of $\FL(x)$ are the same.
        \item Out of every three consecutive symbols of $\FL(x)$ at least one of them is 1.
        \item If $|x|=1$ then $\FL(x)=3$, and otherwise $\FL(x)$ starts by 1 and ends by either 2 or 3.  
    \end{enumerate}
\end{proposition}

The first three items are standard for $R= \log^* |\Gamma| + 10$. The other two can be obtained by a simple modification of the output of the standard function. 
In the output, replace first in parallel each sequence 232 by 212, and then each sequence 323 by 313. This guarantees the fourth condition.
To satisfy the fifth condition, if $|x|=1$, set $\FL(x)=3$, if $|x|=2$, set $\FL(x)=12$, if $|x|=3$, set $\FL(x)=123$, and if $|x|=4$, set $\FL(x)=1212$.
If $|x|>4$ then
replace the sequence at the beginning of the output as follows: if it starts by a word from  $\{2,3\}\{2,3\}1$ replace it by 121,
if it starts by  $\{2,3\}1\{2,3\}\{2,3\}$ replace it by 1212, if it starts by  $\{2,3\}1\{2,3\}1$ replace it by 1231. Then at the end of the sequence, replace $1\{2,3\}1$ by 123,
and $1\{2,3\}\{2,3\}1$ by 1212. This will increase the local dependency to at most $R= \log^* |\Gamma| + 20$.

%In addition to these properties we might also need that, in any string $F(x)$, $x \in \{\Gamma\}^n$, $n\ge 3$, every "3" is adjacent to "1" and "2". If not, then we can recolor "3": if $F(x)$ has "232" as substring, then replace "3" with "1", and if $F(x)$ has "131" as substring, then recplace "3" with "2".

\subsection{Random hash functions}

For sets $U$ and $V$, we say that $\mathcal{H} = \{h:U\rightarrow V\}$ is a {\em pair-wise independent hash system}
if for all $u,u'\in U$ and $v,v'\in V$ if $u\neq u'$ then $\Pr_{h\in \mathcal{H}}[h(u)=v\;\&\;h(u')=v']=\frac{1}{|V|^2}$,
where $h$ is chosen uniformly at random from $\mathcal{H}$.

\begin{proposition}
Let $\mathcal{H} = \{h:U\rightarrow V\}$ be a  pair-wise independent hash system. Let $U' \subseteq U$ where $|U'|=|V|$.
Then for any $v\in V$, $\Pr_{h\in\mathcal{H}}[\exists u\in U', h(u)=v] \ge 1/2$.
\end{proposition}

\begin{proof}
     $\Pr_{h\in\mathcal{H}}[\exists u\in U', h(u)=v] \ge  \sum_{u \in U'} \Pr_{h\in\mathcal{H}}[h(u)=v] - \sum_{ \{u,u'\} \subseteq U'} \Pr_{h\in\mathcal{H}}[h(u)=v\;\&\; h(u')=v] = 1 - \binom{|V|}{2} \cdot \frac{1}{|V|^2} \ge \frac{1}{2}.$
\end{proof}

We will use the following class of randomly selected hash functions to chose the splitting points instead of a fully random function $H$ from $\Gamma^2$ to $\{0,\dots,D-1\}$. 
For integral parameters $D$ and $\ell$, we say that {\em $H : \Gamma^2 \rightarrow \{0,1\}$ is $(D,\ell)$-iterated pair-wise independent function}
if $H$ is obtained by selecting independently at random functions $h_1,\dots,h_\ell : \Gamma^2 \rightarrow \{0,\dots, \ell D -1 \}$
from a pair-wise independent hash system and for each $ab \in \Gamma^2$, $H(ab)$ is set to 0 if $\prod_{i=1}^\ell h_i(ab) = 0$,
and $H(ab)$ is set to 1 otherwise.

Such a hash function $H$ has several useful properties for us: it can be described using $O(\ell \cdot (\log \Gamma + \log D + \log \ell))$ bits,  
at any point it can be evaluated in time polynomial in the bit length of the description of $H$ (so for $\ell = O(\log n)$ and $D$ and $\Gamma$ polynomial in $n$ in time $\OO(1)$),
for any pair of symbols $ab\in \Gamma^2$, the probability that $H(ab)=0$ is roughly $1/D$, and for any sufficiently large set $S \subseteq \Gamma^2$, the image of $S$ under $H$ will contain $0$ with high probability. 
In particular we will use the following simple facts.

\begin{proposition}\label{p-ipw}
Let $H : \Gamma^2 \rightarrow \{0,1\}$ be distributed as $(D,\ell)$-iterated pair-wise independent function.
For any $ab\in \Gamma^2$, $\frac{1}{2D} \le \Pr_H[H(ab)=0]\le \frac{1}{D}$.
Furthermore, for any $S\subseteq \Gamma^2$ where $|S|=\ell D$, $\Pr_H[\forall ab\in S, H(ab)\neq 0]\le 1/2^{\ell}$.
\end{proposition}

\begin{proof}
$\Pr_H[H(ab)=0] = \Pr_{h_1,\dots,h_\ell}[\prod_{i=1}^\ell h_i(ab)=0] \le \sum_{i=1}^\ell \Pr_{h_i}[h_i(ab)=0] \le \ell \cdot \frac{1}{\ell D} = \frac{1}{D}$.
Furthermore,  $\Pr_H[H(ab)\neq 0] = (1-\frac{1}{\ell D})^\ell \le e^{-1/D} \le 1-\frac{1}{2D}$ where we use for $z\in [0,1]$ the inequality $e^{-z} \le 1-\frac{z}{2}$.
For the other claim, by the previous proposition, for a pair-wise independent $h_i$, $\Pr_{h_i} [ \exists ab \in S, h_i(ab)=0] \ge 1/2$.
So $\Pr_{h_i} [ \forall ab \in S, h_i(ab)\neq 0] \le 1/2$. If for all $ab\in S$, $H(ab)\neq 0$ then for all $i\in \{1,\dots,\ell\}$,
for all $ab\in S$, $h_i(ab)\neq 0$. Hence by the independence of $h_1,\dots,h_\ell$,
$\Pr_H[\forall ab\in S, H(ab)\neq 0] \le \Pr_{h_1,\dots,h_\ell} [\bigwedge_{i=1}^\ell \forall ab\in S, h_i(ab)\neq 0]
\le \prod_{i=1}^{\ell}\Pr_{h_i} [\forall ab\in S, h_i(ab)\neq 0] \le 1/2^\ell.$
\end{proof}

\section{Decomposition algorithm}\label{s-decomposition}

%\subsection{Overview}

In this section we describe our main technical tool that we have developed. It is a randomized procedure that splits a string $x$ into blocks $B^x_1,B^x_2,\dots, B^x_s$ and for each block it produces a grammar of size at most $S=\OO(k)$. Furthermore, if $B^x_1,B^x_2,\dots, B^x_s$ is the decomposition 
for a string $x$ and $B^y_1,B^y_2,\dots, B^x_{s'}$ is the decomposition for a string $y$, obtained using the same randomness, where $\ED(x,y)\le k$ then with good probability, $s=s'$ and $B^x_i=B^y_i$ for all but $k$ indices $i$. The edit distance of $x$ and $y$ can be calculated
as $\ED(x,y)=\sum_i \ED(B^x_i,B^y_i)$ where $i$ ranges over the differing blocks.

First we provide an overview of the algorithm, specific details are given in the next sub-section.
The decomposition procedure proceeds in $O(\log n)$ rounds. 
In each round, the algorithm maintains a decomposition of $x$ into {\em compressed} blocks. In each round each block of size at least two is first {\em compressed} and then {\em split}. 
The compression is done by compressing pairs of consecutive symbols into one using a randomly chosen pair-wise independent hash function $C_\ell : \Gamma^2 \rightarrow \Gamma$, where $\ell$ is the round number ({\em level}). 
Non-overlapping pairs of symbols are chosen for compression using a {\em locally consistent coloring} so that every three symbols shrink to at most two. Prior to the compression of pairs we replace each repeated sequence $a^r$ of a symbol $a$, $r\ge 2$, by a special character $\rr_{a,r}$.

The splitting procedure uses a $(D,O(\log n))$-iterated pair-wise independent hash function $H_\ell : \Gamma^2 \rightarrow \{0, 1\}$ to select places where to subdivide each block into sub-blocks, where $D=\OO(k)$ is a suitable parameter. 
We start a new block at each consecutive pair of symbols $ab$, where $H_\ell(ab)=0$.
$H$ is chosen so that for each $ab\in \Gamma^2$, $H(ab)$ happens with probability roughly $1/D$.

After $O(\log n)$ rounds, each block is compressed into at most two symbols and we output a grammar that can generate the block.

For the correctness of the algorithm we will need to establish several properties of the algorithm. Some of these properties are related to behaviour on a single string $x$, others analyze the behaviour of the procedure on a pair of strings $x$ and $y$ of edit distance at most $k$.

The properties we want from the algorithm when it runs on $x$ are the following: In each round, each block should be compressed by factor at least $2/3$ while the size of the required grammar capturing the compression should be $\OO(k)$. 
The former is achieved by the design of the compression procedure. 
The latter goal is provided by the property of the splitting procedure which makes sure that each block $B=b_1b_2\cdot b_m$ resulting from a split has small dictionary $\Dict(B)=\{b_ib_{i+1}, i=1,\dots,m-1\}$. In particular, we require
$|\Dict(B)| = \OO(k)$. The grammar size will be proportional to this dictionary.

For the compression procedure we require that it preserves information so the function $C_\ell$ is one-to-one on each $\Dict(B)$. Since the total size of all dictionaries is bounded by $\OO(n)$ this can be easily achieved by picking $C_\ell$ at random provided that its range size is $\Omega(n^3)$.

Additionally, we need the following property to hold on a pair of strings $x$ and $y$ of edit distance at most $k$ with good probability: The splitting procedure should never split $x$ or $y$ in a {\em region} which is affected by edit operations that transform $x$ to $y$ (for some canonical choice of those operations.) The total size of those regions will be again $\OO(k)$ so we can satisfy this property if each pair of symbols has probability at most $1/\OO(k)$ to start a new block. This constrains the choice of the parameters for the splitting function $H_\ell$.

In the next section we describe the decomposition algorithm fully, and then we establish its properties.

\subsection{Algorithm description}

Let $n$ be an upper bound on the length of the input string and $k\le n$ be given. 
Set $L=\lceil \log_{3/2} n \rceil+3$ to be an upper bound on the decomposition depth.
Let $\Sigma$ be an input alphabet of size at most $n^3$, $\Sigma_c=\{\cc_1,\cc_2,\dots,\cc_{L n^4}\}$ and 
$\Sigma_r=\{\rr_{a,r}, a \in \Sigma \cup \Sigma_c, r \in \{2,3,\dots, n\} \}$ be auxiliary pair-wise disjoint alphabets. 
Let $\Gamma = \Sigma \cup \Sigma_c \cup \Sigma_r$ be the working alphabet, and $\#$ be a symbol not in $\Gamma$. Notice $|\Gamma| = O(n^5 \log n + |\Sigma|)$. 
We call symbols from $\Sigma^0_c=\Sigma$ {\em level-0 compression symbols}, and for $\ell\ge 1$, symbols from $\Sigma_c^\ell=\{\cc_i,\, (\ell-1) n^4 < i \le \ell n^4 \}$ are {\em level-$\ell$ compression symbols}. 
Additionally, symbols from $\Sigma_r^\ell = \{ \rr_{a,r}\in \Sigma_r,$ $a$ is a level-$(\ell-1)$ compression symbol$\}$ are also {\em level-$\ell$ compression symbols}.

Let $R= \log^* |\Gamma| + 20$, $D= 110 R (L+1) k$ and $S=15 DL \log n + 3$ be parameters. 
The algorithm is a recursive algorithm of depth at most $L$. 
It starts by selecting at random several hash functions:
For $\ell=1,\dots, L$, it selects at random a compression hash function $C_\ell:\Gamma^2 \rightarrow \Sigma^\ell_c$ from a pair-wise independent hash family, and for $\ell=0,\dots, L$, it selects at random a splitting function $H_\ell:\Gamma^2 \rightarrow \{0,1\}$ to be a $(D,5\log n)$-iterated pair-wise independent hash function. 
%It selects uniformly at random a length threshold $t\in [10cD \log n, 11cD \log n$. ($c$ will be set later.)

Main building blocks of the algorithm are two functions, $\Compress$ and $\Split$. The first one compresses strings by a factor of $2/3$, and the other splits strings at random points. Their pseudo-code is provided as Algorithm 1 and 2. We describe them next.

$\Compress$. The function $\Compress(B,\ell)$ takes as input a string $B$ over alphabet $\Gamma$ of length at least two, and an integer $\ell \ge 1$, which denotes the level number. Divide $B$ into minimum number of blocks $B_1, \dots, B_m$, $B=B_1B_2B_3\dots B_m$, so that in each $B_i$ either all the characters are the same, i.e. $B_i = a^r$ for some $a\in\Gamma$ and $r\ge 2$, or no two adjacent characters are the same. The first step is to compress the $B_{i}$'s which contain repeated characters by simply replacing the whole $B_i$ with the symbol $\rr_{a,|B_i|}$, where $a$ is the repeated character. Then for the remaining blocks, the following compression is applied: Let $B_i$ be an uncompressed block. Each character of $B_i$ is colored by applying $\FL(B_i)$. Divide $B_i$ into blocks $B_i = B'_1B'_2\dots B'_s$, such that for each $B'_j$ only the first character is colored 1. Now, according to Proposition~\ref{thm:color_reduction}, length of each $B'_j$ is either 2 or 3. If $B'_j = ab$, replace it with $C_{\ell}(ab)$ else if $B'_j = abc$, replace it with $C_{\ell}(ab)\cdot c$, where $a,b,c \in \Gamma$. The actual pseudo-code given below performs the compression of blocks of repeats in two stages, where in the first stage we replace the repeated sequence $a^r$ by $\rr_{a,r} \cdot \#$, and then in the next stage we remove the extra symbol $\#$. This simplifies analysis in Lemma~\ref{l-compress}. Assuming that $C_\ell$ can be evaluated in time $O(1)$, the running time of $\Compress(B,\ell)$ is dominated by the time needed to compute $\FL$-coloring of blocks which is $O(R \cdot |B|)$ in total.

\begin{algorithm}[H]
\label{algo:compress}
%\caption{Compression}
%\DontPrintSemicolon
   \caption{$\Compress(B,\ell)$}\label{alg-compress}
   \KwIn{String $B$ over alphabet $\Gamma$ of length at least two, and level number $\ell$.}
   \KwOut{String $B''$ over alphabet $\Gamma$.}
   
   \vspace{1mm}
   \hrule\vspace{1mm}

    Divide $B=B_1B_2B_3\dots B_m$ into minimum number of blocks so that each maximal subword $a^r$ of $B$, for $a\in\Gamma$ and $r\ge 2$, is one of the blocks.
    
    \For{each $i\in \{1,\dots, m\}$}{
      \lIf{$B_i=a^r$, where $r\ge 2$}{Set $B'_i = \rr_{a,r} \cdot \#$ and color $\rr_{a,r}$ by 1 and $\#$ by 2.\footnotemark}
      \lElse{Set $B'_i = B_i$ and color each symbol of $B'_i$ according to $\FL(B_i)$}
    }
    
    Set $B'=B'_1B'_2\cdots B'_m$, $B''=\eps$, and $i = 1$.

    \While{$i<|B'|$}{

        \lIf{ $B'[i+1] = \#$}{ $B'' = B'' \cdot B'[i]$}
        \lElse{ $B'' = B'' \cdot C_\ell(B'[i,i+1])$}

        $i = i+2$.

        \lIf{$i\le |B'|$ and $B'[i]$ is not colored 1}{
            $B'' = B'' \cdot B'[i]$, $i = i+1$    
        }
    }

    Return $B''$. 

\end{algorithm}
\footnotetext{If $a=\rr_{b,s}$ for some $b\in \Gamma$ and $s\in \N$, then set $B'_i = \rr_{b,{rs}} \cdot \#$. However, such a situation should never happen during the execution of the algorithm as level-$\ell$ compression symbol can be introduced only at level $\ell$.}

$\Split$. The function takes as input a string $B$ over alphabet $\Gamma$ of length at least two, and an integer $\ell \ge 1$. The function splits the string $B$ into smaller blocks.
The algorithm works as follows: 
For each $i  \in \{2,\dots, |B|-1\}$, if $H_\ell(B[i,i+1])=0$, start a new block at position $i$. 
The running time of $\Split(B,\ell)$ is dominated by the time to evaluate $H_\ell$ at $|B|-2$ points.

\begin{algorithm}[H]
   \caption{$\Split(B,\ell)$}\label{alg-split}
   \KwIn{String $B$ over alphabet $\Gamma$ of length at least two, and level number $\ell$.}
   \KwOut{A sequence of strings $(B_0,B_1,\dots,B_s)$ over alphabet $\Gamma$.}
   
   \vspace{1mm}
   \hrule\vspace{1mm}

    Let $i_1< \dots <i_s$ be all $i\in \{2,\dots, |B|-1\}$ where $H_\ell(B[i,i+1])=0$. Set $s=0$ if no such $i$ exists.
    
    Let $i_0=1$ and $i_{s+1}=|B|+1$.
    
    For $j=0,\dots, s$, set $B_j = B[i_j,i_{j+1})$.

    Return $(B_0,B_1,\dots,B_s)$.

\end{algorithm}

The main recursive step of the algorithm is encompassed in function $\Process$. The function gets a block $B\in \Gamma^*$ as its input. 
The block might have already been compressed previously, so the function also gets partial grammars that allow decompression of the block. 
If the block is already of length at most two, then the function outputs the block. 
Otherwise it compresses the block $B$ using $\Compress$, then it subdivides the compressed block using $\Split$, and invokes itself recursively on each sub-block.
For the output, each block is represented by a grammar. The grammar is reconstructed from the compressed block and its partial grammars by a simple bread-first search algorithm provided in the function $\Grammar$.

\begin{algorithm}[H]
   \caption{$\Process(B,(D_1,D_2,\dots,D_{\ell-1}),\ell)$}\label{alg-process}
   \KwIn{String $B \in \Gamma^*$, a sequence of partial grammars $D_i$ over $\Gamma$ for decompressing $B$, and level number $\ell$.}
   \KwOut{A sequence of blocks of $B$ each encoded by a grammar.}
   
   \vspace{1mm}
   \hrule\vspace{1mm}

   \lIf{$|B|\le 2$}{
       Output $\Grammar(B,(D_1,D_2,\dots,D_{\ell-1}),\ell-1)$ and return
   }

   $A=\Compress(B,\ell)$.

   $(B_0,B_1,\dots,B_s) = \Split(A,\ell)$.

   $D_\ell = \{C_\ell(ab)\rightarrow ab;\; ab\in \Dict(B)\}$
   
   For $i=0,\dots, s$, $\Process(B_i,(D_1,\dots,D_{\ell-1},D_\ell),\ell+1)$.
\end{algorithm}

To decompose an input string $x$ into blocks, we first apply function $\Split(x,0)$ to $x$ and then invoke $\Process(B,(),1)$ on each of the obtained blocks $B$. 
Breaking the string $x$ into sub-blocks guarantees that each block passed to $\Process$ has small dictionary whereas the dictionary of $x$ could have been arbitrarily large.

\begin{algorithm}[H]
\label{alg-prune}
   \caption{$\Grammar(B,(D_1,D_2,\dots,D_{\ell}),\ell)$}
    \KwIn{String $B \in \Gamma^*$, a sequence of partial grammars $D_i$ over $\Gamma$ for decompressing $B$.}
   \KwOut{The smallest grammar $G$ for $B$ based on the grammars $D_i$.}
   
   \vspace{1mm}
   \hrule\vspace{1mm}

    Let $C = \{c \in \Sigma_c: c \text{ appears in }B\text{ or }\rr_{c,r}  \text{ appears in }B \text{ for some }r\}$. {\em \hfil // Symbols needed to decompress $B$}
    
    $G = \{ \# \rightarrow B\}$.
    
    \For{ $j=\ell,\dots, 1$}
    {
        \For{each $c \in C$}{
            \lIf{$c\rightarrow ab \in D_j$ for some $ab\in \Gamma^2$}{
                $G =G \cup \{ c \rightarrow ab\}$,

                \Indp$C = C \cup \{c' \in \Sigma_c; c' \in \{a,b\} \text{ or }\rr_{c',r} \in \{a,b\} \text{ for some }r\}$
            }
        }
    }

    For each $\rr_{a,r}$ appearing in any of the rules in $G$, add $\rr_{a,r} \rightarrow a^r$ to $G$.

    Return $G$.

\end{algorithm}

\subsection{Correctness of the decomposition algorithm}\label{s-correctness}

Our goal is to establish the following theorem which is a stronger version of Theorem~\ref{t-main1}:

\begin{theorem}\label{t-decomposition}
    Let $k\le n$ be integers.
    Let $x$ and $y$ be a pair of strings of length at most $n$ with $\ED(x,y) \le k$. Let $G^x_1,\dots,G^x_s$ and  $G^y_1,\dots,G^y_{s'}$ 
    be the sequence of grammars output by the decomposition algorithm on input $x$ and $y$ respectively, using the same choice of random functions
    $C_1,\dots,C_L$ and $H_0,\dots,H_L$. The following is true for $n$ large enough:
    \begin{enumerate}
        \item With probability at least $1-2/n$, $x=\eval(G^x_1)\cdots \eval(G^x_s)$ and $y=\eval(G^y_1)\cdots \eval(G^y_{s'})$.
        \item With probability at least $1-2/n$, for all $i \in \{1,\dots,s\}$ and $j\in \{1,\dots, s'\}$, $|G^x_i|,|G^y_j| \le S$.
        \item With probability at least $9/10$, $s=s'$, $G^x_i=G^y_i$, for all $i\in \{1,\dots,s\}$ except for at most $k$ indices $i$, and
         $\ED(x,y)=\sum_i \ED(\eval(G^x_i),\eval(G^y_i))$.
    \end{enumerate}
\end{theorem}

By union bound, all three parts happen simultaneously with probability at least $9/10-4/n$ which is $\ge 4/5$ for $n$ large enough.

 To prove the theorem we make some simple observations about the algorithm, first.

\begin{lemma}
    For any string $B$ of length at least two, and $\ell \ge 1$,   $|\Compress(B,\ell)| \le \frac{2}{3}|B|+1$ and $|\Compress(B,\ell)|<|B|$.
\end{lemma}

\begin{proof}
Let $B=B_1B_2B_3\dots B_m$ be as in the procedure. 
Every block $B_i$ that equals to $a^r$, for some $a$ and $r\ge 2$, is reduced to one symbol by the compression. 
The other blocks are colored using $\FL(\cdot)$ and compressed. 
Unless a block $B_i$ is of size one, the coloring induces division of the block $B_i$ into subwords of size two or three, 
where the former is compressed into one symbol and the latter into two symbols.
Hence, each such a block is compressed to at most $2/3$ of its size.
So the only blocks $B_i$ that do not shrink are of size one, and are sandwiched between blocks of repeated symbols 
(that shrink by a factor of at least two). 
The worst-case situation is when $m$ is odd, blocks $B_i$  are of size one for odd $i$,
and of size two for even $i$. In that case the original string $B$ shrinks to size $\lfloor \frac{2}{3}|B| \rfloor + 1$. 
This proves the first inequality.
The second inequality is also clear from the analysis above: The only time the string does not shrink is if it is of size one.
\end{proof}

\begin{corollary}\label{c-decomposition1}
   On a string $B$ of length at most $n$, the depth of the recursive calls of $\Process$ is at most $L$.
\end{corollary}

Indeed, from the previous lemma it follows that each block after $\ell$ compressions and splits is of size at most $(2/3)^\ell |B| + 3$.
Hence, after $L=\lceil \log_{3/2} n \rceil+3$ recursive calls $\Process$ must stop the recursion.

\begin{lemma}\label{l-decomposition1}
    Let $B\in \Gamma^*$ be of length at most $n$, and $\ell \in \{0,\dots,L\}$. Let $(B_0,B_1,\dots, B_s)=\Split(B,\ell)$ where $H_\ell : \Gamma^2 \rightarrow \{0,\dots,D-1\}$ is chosen to be a random $(D, 5\log n)$-iterated pair-wise independent hash function. Then with probability at least $1-1/n^3$, for all $j\in \{0,\dots,s\}$, $|\Dict(B_j)| \le 5D \log n$.
\end{lemma}

\begin{proof}
If for some $j\in \{0,\dots,s\}$, $|\Dict(B_j)| > 5D \log n$, then there exists $1 < r < t \le |B|$ such that $|\Dict(B[r,t])| = 5D \log n$
and for all $i\in \{r,\dots,t-1\}$, $H_\ell(B[i,i+1]) \neq 0$. 
(Pick $r$ to be the position in $B$ of the second symbol of $B_j$ and $r$ some later position in $B_j$.) 
For a fixed $r$ and $t$ with $|\Dict(B[r,t])| = 5D \log n$, $\Pr_{H_\ell}[ \forall i \in \{r,\dots,t-1\}, H_\ell(B[i,i+1])\neq 0] \le 2^{-5\log n}$ by Proposition~\ref{p-ipw}. 
Hence, $\Pr_{H_\ell}[ \exists 1<r<t \le |B|, |\Dict(B[r,t])| = 5D\log n \text{ and } \forall i \in \{r,\dots,t-1\}, H_\ell(B[i,i+1])\neq 0] \le |B|^2 \frac{1}{n^5} \le n^2/ n^{5} \le 1/n^3.$
\end{proof}

\begin{lemma}
    For $B\in \Gamma^*$, $\ell \le L$, $D_1,\dots,D_\ell$ be partial grammars over $\Gamma$, $\Grammar(B,(D_1,\dots,D_\ell),\ell)$ outputs a grammar $G$ of size at most $1+|B|+ 3\sum_i |D_i|$, and runs in time $\OO(\ell(|B|+|G|))$.
\end{lemma}

\begin{proof}
First we add the starting rule to $G$.
Then in each iteration of the main loop we can add a rule of the type $c \rightarrow ab$ to $G$ from some $D_j$.
Hence, the number of such rules in $G$ is at most $1+\sum_j |D_j|$. 
Last, we add to $G$ rules for symbols from $\Sigma_r$ that appear on right hand sides of rules in $G$.
This increases the size of $G$ by at most $|B|+2\sum_j |D_j|$.
If $D_j$'s are stored using some efficient data structure such as binary search trees or hash tables indexed by the left hand side of rules, finding and adding each new rule to $G$ takes time $\OO(1)$.
The size of $C$ is bounded by $O(|B|+|G|)$ so the nested loops make at most $O(\ell(|B|+|G|))$ iterations in total.    
%(We assume that evaluation of $C_j(\cdot)$ takes $O(1)$.) 
Hence, the total running time is bounded as claimed.
\end{proof}

During processing of a string $x$, there are at most $Ln$ calls to the function $\Split$. 
(The actual number of calls is $O(n)$ as the strings shrink exponentially but our simple upper bound suffices.) 
The probability that any one of them would produce a block with dictionary larger than $5D \log n$ is at most $Ln/n^3$. 
If all dictionaries are of size at most $5D \log n$ then so are all the partial grammars produced by $\Process$.
We can conclude the next corollary which implies the second item of Theorem~\ref{t-decomposition}.

\begin{corollary}\label{c-part2}
For $n$ large enough, on a string $x$ of length at most $n$, 
processing the string $x$ produces a sequence of grammars each of size at most $S= 15DL \log n + 3$ 
with probability at least $1-1/n$.
\end{corollary}

For the grammars produced by the algorithm to be deterministic, we need that each $C_\ell$ is one-to-one on $\Dict(B)$
for each block $B$ on which $\Compress(B,\ell)$ is invoked. That will happen with high probability by a standard argument:

\begin{lemma}\label{l-onetoone}
    Let $B\in \Gamma^*$ be of length at most $n$ and $\ell \in \{1,\dots,L\}$. 
    Let $C_\ell:\Gamma^2 \rightarrow \{\cc_i,\, (\ell-1) n^4 < i \le \ell n^4 \}$ be chosen at random from a pair-wise independent family of hash functions. 
    Then with probability at least $1-|B|/n^3$, $C_\ell$ is one-to-one on $\Dict(B)$.
\end{lemma}

\begin{proof}
For two distinct elements from $\Dict(B)$, the probability of a collision for randomly chosen $C_\ell$ is at most $1/n^4$.
By the union bound, the probability that $C_\ell$ is not one-to-one on $\Dict(B_j)$ is at most $|\Dict(B)|^2/n^4 \le |B|/n^3$ as $|\Dict(B)|\le |B|\le n$.
\end{proof}

During processing of a string $x$, there are at most $Ln$ calls to the function $\Compress$. 
For a fixed level $\ell \in \{1,\dots,L\}$, the total size of blocks $B$ for which $\Compress(B,\ell)$ is invoked
is at most $n$. 
By the previous lemma and the union bound, the probability that during any of those calls $\Compress(B,\ell)$ uses 
a function $C_\ell$ that is not one-to-one on $\Dict(B)$ is at most $1/n^2$.
If all the hash functions $C_1, C_2,\dots, C_L$ that are used to compress blocks of $x$ are one-to-one on their respective blocks then the grammars that $\Grammar$ produces will be deterministic, and they will evaluate to their respective blocks of $x$.
(We can actually conclude a stronger statement that each $C_\ell$ will be one-to-one on the union of all blocks at level $\ell$ with high probability.)
We can conclude the next corollary which implies the first item of Theorem~\ref{t-decomposition}.

\begin{corollary}\label{c-part1}
For $n$ large enough, on a string $x$ of length at most $n$, with probability at least $1-L/n^2$,
processing the string $x$ produces a sequence of grammars $G_1,G_2,\dots, G_s$ 
such that $x=\eval(G_1) \cdots \eval(G_s)$.
\end{corollary}

At this point we can estimate the running time of the decomposition algorithm. 
We can let the algorithm fail, and produce some trivial decomposition of $x$, 
whenever $\Split$ produces a block with dictionary larger than $5D \log n$. 
If it does not fail, then all grammars are of size at most $S$ which is $\OO(k)$. 
There are at most $|x|$ of them and their total size is at most $2|x|$ as each of the grammars $G$ produces a string of size at least $|G|/2$.
So time spent in $\Grammar(\dots)$ is bounded by $\OO(|x|)$. 
The total time spent in $\Compress(\dots)$ is proportional to the sum of sizes of all non-trivial blocks over all levels
of recursion which is $O(|x|L)=\OO(|x|)$. 
(A more accurate estimate on the total size of blocks is $O(|x|)$ since the blocks are shrinking geometrically in each iteration.)
This means that the time to execute all calls to $\Compress$ is $O(|x|L R) = \OO(|x|)$.
The time spent in $\Split(\dots)$ is dominated by the time needed to evaluate $H_\ell$. 
The number of evaluation points at a given level $\ell$ is proportional to the total size of all blocks at that level.
Since $H_\ell$ can be evaluated at a single point in time $\OO(1)$, we get an upper bound $O(|x|L)=\OO(|x|)$ on time spent in $\Split$.
Hence, in total the decomposition procedure runs in time $\OO(|x|)$.
% (We believe that the total running time can be improved to $\OO(n)$ on average.
% One could argue that in expectation the number of grammars the procedure produces is $\OO(n/k)$ 
% as the average block size a string $x$ is decomposed into should be at least $\Omega(D/\log n)$. 
% So we believe that the total running time of calls to $\Grammar$ is $\OO(n)$.) 

\begin{proposition}
Given $k\le n$, the running time of the decomposition algorithm on a string $x$ of length at most $n$ is $\OO(|x|)$ with probability at least $1-1/n$. 
\end{proposition}

It remains to address the properties of the algorithm run on a pair of strings $x$ and $y$ of edit distance at most $k$ to establish Theorem~\ref{t-decomposition}.
For the pair of strings $x$ and $y$ we fix a {\em canonical decomposition of  $x$ and $y$} to be a sequence of words $w_0,w_1,\dots,w_k, u_i,\dots, u_k, v_1, \dots, v_k \in \Gamma^*$ such that $x=w_0 u_1 w_1 u_2 w_2\cdots u_k w_k$, $y=w_0 v_1 w_1 v_2 w_2 \cdots v_k w_k$
and $|u_i|,|v_i| \le 1$ for all $i$. By the definition of edit distance such a decomposition exists: each pair $(u_i, v_i)$ represents one edit operation, and we fix one such decomposition to be {\em canonical}.
Observe, if we now partition $x$ into blocks $B^x_1,\dots,B^x_s$ so that each $B^x_i$ starts within one of the $w_j$'s, 
and we partition $y$ into blocks $B^y_1,\dots,B^y_s$ so that each block $B^y_i$ starts at the corresponding location in $w_j$ as $B^x_i$, 
then $\ED(x,y)=\sum_i \ED(B^x_i,B^y_i)$.

We need to understand what happens with the decomposition of $x$ and $y$ when we apply the $\Compress$ function.
Let $x=uwv$ and $x'=\Compress(x,\ell)=u'w'v'$, for some $u,w,v,u'w'v' \in \Gamma^*$. 
We say that a symbol $c$ in $w'$ {\em comes from the compression of $w$} 
if either it is directly copied from $w$ by $\Compress$, or 
it is the image $c=C_\ell(ab)$ of a pair of symbols $ab$ where $a$ belongs to $w$, or 
$c=\rr_{a,r}$ replaced a block $a^r$ where the first symbol of $a^r$ belongs to $w$.  
{\em $w'$ is the compression of $w$} if it consists precisely of the symbols that come from the compression of $w$. 
Furthermore, we say a symbol $c$ in $w'$ {\em comes weakly from the compression of $w$} 
if either it is directly copied from $w$ by $\Compress$, or 
it is the image $c=C_\ell(ab)$ of a pair of symbols $ab$ where $a$ or $b$ belong to $w$, or 
$c=\rr_{a,r}$ replaced a block $a^r$ where some symbol of $a^r$ belongs to $w$.
{\em $w'$ is the weak compression of $w$} if it consists precisely of the symbols that come weakly from the compression of $w$. 
Notice, a weak compression of $w$ might contain and extra symbol at the beginning compared to the compression of $w$.

The following lemma captures what compression does to the canonical decomposition of $x$ and $y$. (See Fig.~\ref{fig:compression} for illustration.)

\begin{lemma}\label{l-compress}
Let $x$ and $y$ be strings over $\Gamma$, and let $x'=\Compress(x,\ell)$ and $y'=\Compress(y,\ell)$.
Let $x=w_0 u_1 w_1 u_2 w_2\cdots u_q w_q$ and $y=w_0 v_1 w_1 v_2 w_2 \cdots v_q w_q$ for some strings $w_i$, $u_i$ and $v_i$ where for $i\in \{1,\dots,q\}$, $|u_i|,|v_i| \le 4R + 24$. 
Then there are  $w'_0,w'_1, \dots, w'_q,u'_1, \dots, u'_q,v'_1, \dots, v'_q \in \Gamma^*$ such that 
for $i\in \{1,\dots,q\}$, $|u'_i|,|v'_i| \le 4R + 24$, $x'=w'_0 u'_1 w'_1 u'_2 w'_2\cdots u'_q w'_q$ 
and $y'=w'_0 v'_1 w'_1 v'_2 w'_2\cdots v'_q w'_q$.
Moreover, each $w'_i$ is the compression of the same subword of $w_i$ in both $x$ and $y$.
\end{lemma}

For each $x=w_0 u_1 w_1 u_2 w_2\cdots u_q w_q$, $y=w_0 v_1 w_1 v_2 w_2 \cdots v_q w_q$ and $\ell$ we fix one choice of $w'_0,\dots, w'_q,$ $u'_0, \dots, u'_q,$ $v'_0,\dots, v'_q$ satisfying the lemma. We will refer to it as the {\em canonical decomposition} of $x'$ and $y'$ induced by the decomposition of $x$ and $y$ as given by the lemma.

\begin{proof}
The first stage of $\Compress$ replaces maximal blocks of repeated symbols by shortcuts. To simplify our analysis first we will reassign blocks of repeated symbols
among neighboring blocks of $w_i$, $u_i$ and $v_i$, resp., so each maximal block of symbols in $x$ and $y$ is fully contained in one of the words $w_i$, $u_i$ or $v_i$.

For $i=1,\dots, q-1$ we define words $w^{(1)}_i$ and parameters $a_i, b_i \in \Gamma$ and $k_i, k'_i \in \N$ as follows: If $w_i$ contains at least two distinct symbols 
let $w_i=a_i^{k_i} w^{(1)}_i b_i^{k'_i}$ so that $k_i$ and $k'_i$ are maximum possible, otherwise $w_i=a_i^{k_i}$ for some $a_i$ and $k_i$ ($k_i$ might be zero), and we set $w^{(1)}_i = \eps$, $b_i=a_i$ and $k'_i=0$. 
Let $w_0=w^{(1)}_0 b_0^{k'_0}$ for maximum possible $k'_0$ and some symbol $b_0$.
Let $w_q= a_q^{k_q} w^{(1)}_q$ for maximum possible $k_q$ and some symbol $a_q$. 
For $i=1,\dots,q$, we let $u^{(1)}_i = b^{k'_{i-1}}_{i-1} u_i a^{k_i}_i$.
Similarly,  $v^{(1)}_i = b^{k'_{i-1}}_{i-1} v_i a^{k_i}_i$. Hence, $x=w^{(1)}_0 u^{(1)}_1 w^{(1)}_1 \cdots u^{(1)}_q w^{(1)}_q$ and $y=w^{(1)}_0 v^{(1)}_1 w^{(1)}_1  \cdots v^{(1)}_q w^{(1)}_q$. 

Next, if there is a maximal block of symbols $a^r$ contained in $u^{(1)}_s w^{(1)}_s \cdots u^{(1)}_t$ starting in $u^{(1)}_s$ and ending in $u^{(1)}_t$, $s\neq t$, we add all the symbols of the $a^r$ to the end of $u^{(1)}_s$ and remove them from the other $u^{(1)}_i$, $i=s+1,\dots,t$. (Notice, $w^{(1)}_i=\eps$ for $s<i<t$ because of the definition of $w^{(1)}_i$, and $u^{(1)}_i$ will become empty for $s<i<t$.) We do this for all maximal blocks of repeated symbols that span multiple $u^{(1)}_i$. We perform similar moves on $v^{(1)}_i$'s. 
After all of those moves we denote the resulting subwords by $w^{(2)}_i$, $u^{(2)}_i$, and $v^{(2)}_i$. (Notice, $w^{(2)}_i=w^{(1)}_i$ for all $i$.)
We have: $x=w^{(2)}_0 u^{(2)}_1 w^{(2)}_1 \cdots u^{(2)}_q w^{(2)}_q$ and $y=w^{(2)}_0 v^{(2)}_1 w^{(2)}_1  \cdots v^{(2)}_q w^{(2)}_q$.
At this stage, each maximal block of repeated symbols in $x$ or $y$ is contained in one of the subwords $w^{(2)}_i$, $u^{(2)}_i$, and $v^{(2)}_i$.

The first stage of $\Compress$ replaces each maximal block $a^r$, $r\ge 2$, by a sequence $\rr_{a,r} \#$, 
and we apply this procedure on each subword $w^{(2)}_i$, $u^{(2)}_i$, and $v^{(2)}_i$ to obtain corresponding subwords $w^{(3)}_i$, $u^{(3)}_i$, 
and $v^{(3)}_i$. 
Observe, for $i=1,\dots,q$, $|u^{(3)}_i|,|v^{(3)}_i| \le 4R + 28$. 
This is because every $u_i$ is transformed into $u^{(3)}_i$ by appending or prepending possibly empty block of repeated symbols, i.e., $u^{(3)}_i =  a^r u_i b^{r'}$ for some $a,b,r,r'$, or removing its content entirely. 
Each block of repeats is reduced to two symbols so each $u^{(3)}_i$ is longer than the original by at most 4 symbols. Similarly for $v^{(3)}_i$.

Next, coloring function $\FL$ is used on parts of $x$ and $y$ that are not obtained from repeated symbols; the two symbols replacing each repeated block are  colored by $1$ and $2$, resp. 
We refer to this as $\{1,2,3\}$-coloring. At most $R$ first and last symbols of each $w^{(3)}_i$ might be colored differently in $x$ and $y$ as the color of each symbol depends on the context of  at most $R$ symbols on either side of the symbol, and that context might differ in $x$ and $y$. 
Hence, only symbols near the border of $w^{(3)}_i$ that are in vicinity of $u^{(3)}_i$'s and $v^{(3)}_i$'s, resp., might get different colors. 
All the other symbols of $w^{(3)}_i$ are colored the same in both $x$ and $y$. 
The coloring is then used to make decisions on which pairs of symbols are compressed into one.

We will let $u'_i$ be the symbols that come from the compression of symbols in $u^{(3)}_i$, 
the first up-to $R+2$ symbols of $w^{(3)}_i$, and the last up-to $R+3$ symbols of $w^{(3)}_{i-1}$. 
Next we specify precisely which symbols of  $w^{(3)}_i$ and $w^{(3)}_{i-1}$ are considered to be compressed into symbols belonging to $u'_i$.   
For $i=0,\dots, q$, if $|w^{(3)}_i| \ge R+3$, let $s^x_i$ be the position of the first symbol in $w^{(3)}_i$ among positions $R+1,R+2,R+3$ which is colored 1 in $x$ by the $\{1,2,3\}$-coloring. 
If  $|w^{(3)}_i| < R+3$, let $s^x_i=1$. 
Next, if $|w^{(3)}_i| \ge 2R+3$ set $t^x_i$ to be the first position from left colored 1 among the symbols of $w^{(3)}_i$ at positions $R+1, R+2, R+3$ counting from right. 
If $|w^{(3)}_i| < 2R+3$, set $t^x_i$ to be equal to $s^x_i$.
For $i=0$, if $|w^{(3)}_0| \ge R+3$ then redefine $s^x_0=1$.
For $i=q$, redefine $t^x_q=|w^{(3)}_q|+1$ and if $|w^{(3)}_q| < R+3$ then redefine $s^x_q$ to $t^x_q$.
Similarly, define $s^y_i$ and $t^y_i$ based on the $\{1,2,3\}$-coloring of $y$. 

Notice, $s^x_i\neq t^x_i$ iff  $s^y_i\neq t^y_i$. 
Furthermore, if  $s^x_i\neq t^x_i$ then either $i \in \{q,0\}$ or $|w^{(3)}_i| \ge 2R+3$ so $s^x_i=s^y_i$ and $t^x_i=t^y_i$ as the symbols $R$-away from either end of $w^{(3)}_i$ are colored the same in $x$ and $y$.
We let $u'_i$ to be the compression of  $w^{(3)}_{i-1}[t^x_{i-1},|w^{(3)}_{i-1}|] \cdot u^{(3)}_i \cdot w^{(3)}_i[1,s^x_i)$ and 
similarly, $v'_i$ to be the compression of  $w^{(3)}_{i-1}[t^y_{i-1},|w^{(3)}_{i-1}|] \cdot v^{(3)}_i \cdot w^{(3)}_i[1,s^y_i)$. 
We let $w'_i$ be the compression of $w^{(3)}_i[s^y_i,t^y_i)$. 

Hence, $u'_i$ comes from the compression of at most $|u^{(3)}_i| + 2R+5 \le 6R + 33$ symbols. 
Since each symbol after a symbol colored 1 is {\em removed} by the compression, and each consecutive triple of symbols contains at least one symbol colored by 1, the at most $6R + 27$ symbols are compressed into at most $(6R + 33)\cdot 2/3 + 2 =4R+24$ symbols. So $u'_i$ is of length at most $4R+24$. Similarly for $v'_i$.
\end{proof}

The following generalization of the previous lemma will be useful to design a rolling sketch. 
It considers situation where $x$ and $y$ are prefixed by some strings $u$ and $v$, resp., that we want to ignore from the analysis. The proof of the lemma is a straightforward modification of the above proof.

\begin{lemma}\label{l-compress-rolling}
Let $x,y,u,v\in \Gamma^*$, and let $u'x'=\Compress(ux,\ell)$ and $v'y'=\Compress(vy,\ell)$, where
$x'$ is the weak compression of $x$, and $y'$ is the weak compression of $y$.
Let $x=u_0w_0 u_1 w_1 u_2 w_2\cdots u_q w_q$ and $y=v_0w_0 v_1 w_1 v_2 w_2 \cdots v_q w_q$ for some strings $w_i$, $u_i$ and $v_i$ where for $i\in \{0,\dots,q\}$, $|u_i|,|v_i| \le 4R + 24$. 
Then there are  $w'_0,w'_1, \dots, w'_q,u'_0,u'_1, \dots, u'_q,v'_0,v'_1, \dots, v'_q \in \Gamma^*$ such that 
for $i\in \{0,\dots,q\}$, $|u'_i|,|v'_i| \le 4R + 24$, $x'=u'_0w'_0 u'_1 w'_1 u'_2 w'_2\cdots u'_q w'_q$ 
and $y'=v_0'w'_0 v'_1 w'_1 v'_2 w'_2\cdots v'_q w'_q$. 
Moreover, each $w'_i$ is the compression of the same subword of $w_i$ in both $x$ and $y$.
\end{lemma}

Let $x\in \Sigma^*$. Let $H_0,H_1,\dots,H_L, C_1,C_2,\dots,C_L$ be chosen. 
We define inductively the {\em trace} of the algorithm on $x$ at level $\ell\ge 0$ to consist of sequences 
$B^{x}(\ell,1),\dots, B^{x}(\ell,s^{x}_{\ell}) \in \Gamma^*$, of auxiliary sequences $A^{x}(\ell,1),\dots, A^{x}(\ell,{s^{x}_{\ell}}) \in \Gamma^*$ and $t^{x}_{\ell,1},\dots, t^{x}_{\ell,{s^{x}_{\ell}+1}} \in \N$. 
Their meaning is: $B^x(\ell,i)$ is compressed into $A^x(\ell,i)$ and that is split into blocks $B^x(\ell+1,j)$ for $t^x_{\ell+1,i} \le j < t^x_{\ell+1,i+1}$. (See Fig.~\ref{fig:hierarchy} for illustration.)\footnote{To avoid double and triple indexes we use our notation $B^x(\ell,i)$ and $A^x(\ell,i)$ instead of the usual  $B^x_{\ell,i}$ and $A^x_{\ell,i}$.}

Set $$B^{x}(0,1),\dots, B^{x}(0,{s^x_0})=\Split(x,0).$$
For $\ell=1,\dots,L$ we define $B^{x}(\ell,1),\dots, B^{x}(\ell,s^x_\ell)$ inductively. Set $t^{x}_{\ell,1}=1$.
For $i=1,\dots, s^{x}_{\ell-1}$, if $|B^{x}(\ell-1,i)| \ge 2$, then
$$A^{x}(\ell-1,i) = \Compress(B^{x}(\ell-1,i),\ell),$$ 
and for $(B_0,B_1,\dots,B_s)= \Split(A^{x}(\ell-1,i),\ell)$ set
$$B^{x}(\ell,{t^{x}_{\ell,i}}) = B_0,\,\,\, B^{x}(\ell,{t^{x}_{\ell,i}+1}) = B_1,\,\,\, \dots,\,\,\, B^{x}(\ell,{t^{x}_{\ell,i}+s}) = B_s$$
and $t^{x}_{\ell,{i+1}} = t^{x}_{\ell,i} + s +1$.
If $|B^{x}(\ell-1,i)| < 3$, then set
$B^{x}(\ell,{t^{x}_{\ell,i}})$ and $A^{x}(\ell-1,i)$ to  $B^{x}(\ell-1,i),$ and $t^{x}_{\ell,{i+1}} = t^{x}_{\ell,i} + 1$.
For $j=s^{x}_{\ell-1}$, set $s^{x}_{\ell} = t^{x}_{\ell,j+1}$.

\begin{figure}[htp]
    \centering
    \includegraphics[width=\textwidth]{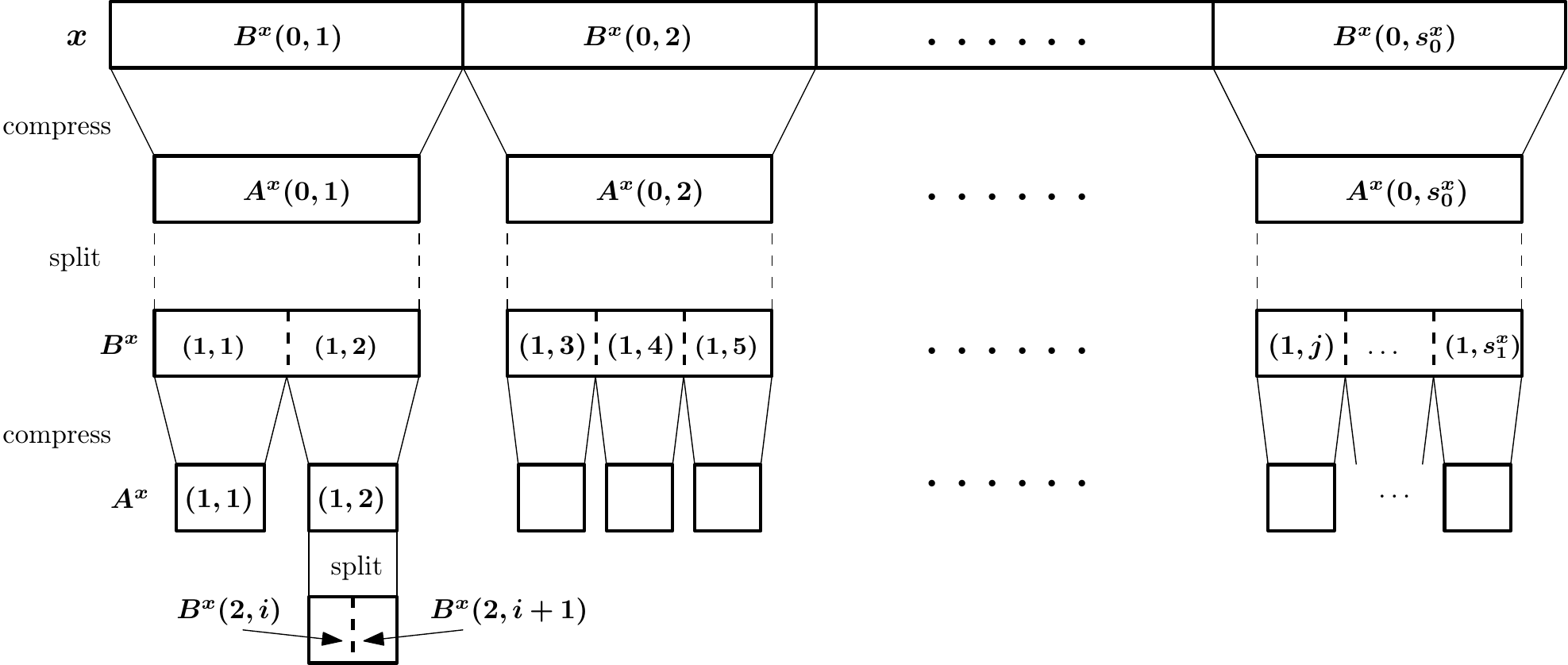}
    \caption{The hierachical decomposition of $x$.}
    \label{fig:hierarchy}
\end{figure}

\medskip
Furthermore, for $x$ and $y\in \Sigma^*$, $\ell, i\ge 0$, define a canonical decomposition of blocks $A^x(\ell,i),$ $B^x(\ell,i),$ $A^y(\ell,i),$ $B^y(\ell,i)$ inductively as follows.
Let $A^x(-1,1)=x$ and $A^y(-1,1)=y$. Let $t^x_{-1,1}=1, t^x_{-1,2}=2,$ $s^x_{-1}=1$, $t^y_{-1,1}=1,$  $t^y_{-1,2}=2$, and $s^y_{-1}=1$.
Let $$A^x(-1,1)=w_0 u_1 w_1 u_2 w_2\cdots u_k w_k \,\,\,\,\&\,\,\,\, 
A^y(-1,1)=w_0 v_1 w_1 v_2 w_2 \cdots v_k w_k$$ be the canonical decomposition of the pair $x$ and $y$. 

For $\ell \ge 0$ and $j\in \{1,\dots, s^{x}_{\ell}\}$, let $i$ be such that 
$t^x_{\ell-1,i} \le j < t^x_{\ell-1,i+1}$ and $m=j-t^x_{\ell-1,i}$.
Then $B^x(\ell,j)$ is the $m$-th block of $\Split(A^x(\ell-1,i),\ell)$.
If the decomposition of $A^x(\ell-1,i)$ is defined and is equal to $w_0 u_1 w_1 u_2 w_2\cdots u_q w_q$, for some $u_i,w_i\in \Gamma^*$, then 
the decomposition of $B^x(\ell,j)$ is the restriction of the decomposition of $A^x(\ell-1,i)$
to symbols of the $m$-th block of $\Split(A^x(\ell-1,i),\ell)$. Otherwise the decomposition of $B^x(\ell,j)$ is undefined. Similarly for $B^y(\ell,j)$. (See Fig.~\ref{fig:compression}.)

For $\ell \ge 0$ and $j\in \{1,\dots, s^{x}_{\ell}\}$, if $B^x(\ell,j)$ and $B^y(\ell,j)$ have defined decompositions $B^x(\ell,j)= w_0 u_1 w_1 u_2 w_2\cdots u_q w_q$ and $B^y(\ell,j)=w_0 v_1 w_1 v_2 w_2\cdots v_q w_q$ for some $u_i,v_i,w_i \in \Gamma^*$,
then we let $A^x(\ell,j) = w'_0 u'_1 w'_1 u'_2 \cdots w'_{q}$ 
and $A^y(\ell,j) = w'_0 v'_1 w'_1 v'_2 \cdots w'_{q}$ be their canonical decomposition induced by $B^x(\ell,j)$ and $B^x(\ell,j)$ as given by Lemma~\ref{l-compress}.

\begin{figure}[htp]
    \centering
    \includegraphics[width=\linewidth]{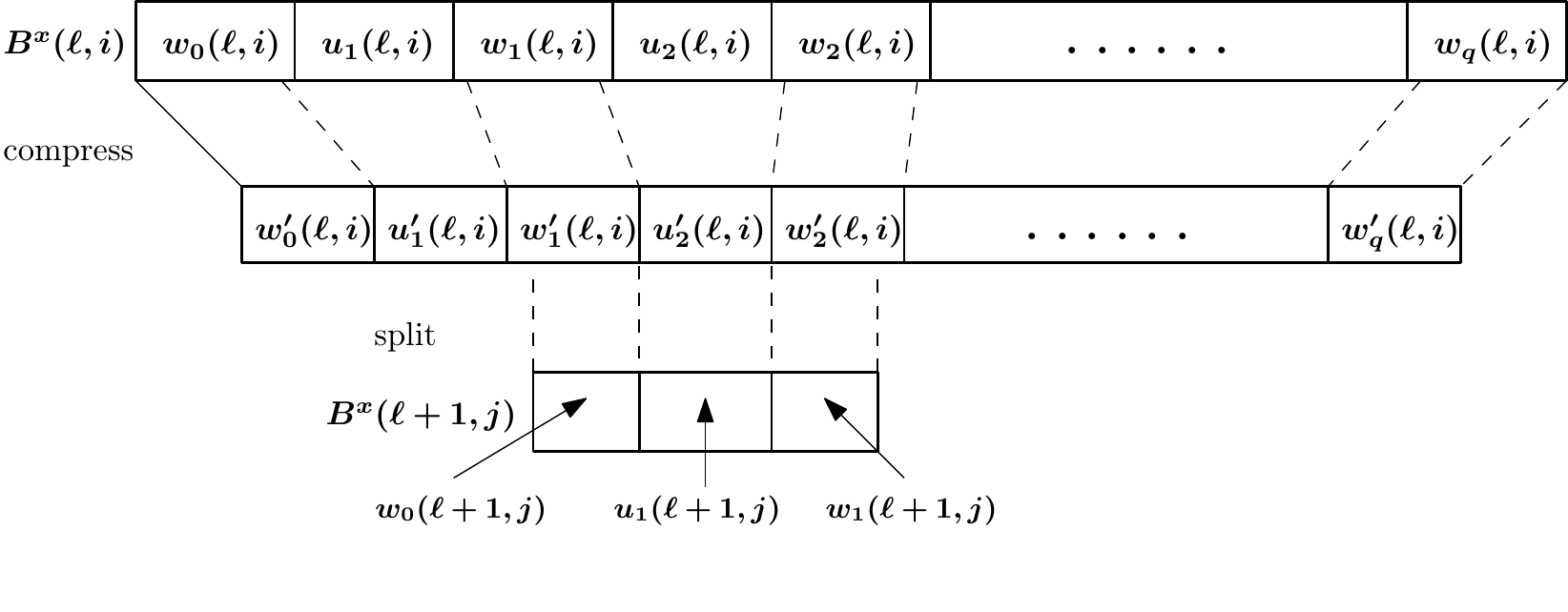}
    \caption{Decomposition of $B^x(\ell,i)$ after compression and split.}
    \label{fig:compression}
\end{figure}

\bigskip
To conclude item 3 of Theorem~\ref{t-decomposition} we want to argue that $x$ and $y$ are recursively split into sub-blocks that respect their canonical decomposition. So we want all splits of blocks to occur in matching parts of $x$ and $y$. For $A^x(\ell-1,i)$ with canonical decomposition
$w_0 u_1 w_1 u_2 w_2\cdots u_q w_q$ we say that $\Split(A^x(\ell-1,i), \ell)$ makes {\em undesirable split} if it starts a new block at a position $j$
that either belongs to one of the $u_1,u_2,\dots,u_q$ or is the first or last symbol of one of the $w_0,w_1,\dots,w_q$. 
Recall, $\Split(A^x(\ell-1,i), \ell)$ starts a new block at each position $j$ such that $H_\ell(A^x(\ell-1,i)[j,j+1])=0$. 
Since $H_\ell$ is chosen at random a given position starts a new block with probability $1/D$.

Similarly, for $A^y(\ell-1,i)$ with canonical decomposition
$w'_0 v_1 w'_1 v_2 \cdots v_{q'} w'_{q'}$ we say that $\Split(A^y(\ell-1,i), \ell)$ makes {\em undesirable split} 
if it starts a new block at a position $j$ that either belongs to one of the $v_1,v_2,\dots,v_{q'}$ or is the first or last symbol of one of 
the $w'_0,w'_1,\dots,w'_{q'}$. 
If $A^x(\ell-1,i)$ and $A^y(\ell-1,i)$ have {\em matching} canonical decomposition (that is $q=q'$ and each $w_j=w'_j$) 
and both $\Split(A^x(\ell-1,i), \ell)$ and $\Split(A^y(\ell-1,i), \ell)$ make no undesirable split then $A^x(\ell-1,i)$ and $A^y(\ell-1,i)$ 
are split in the same number of blocks with matching canonical decomposition as they are split at the same positions in the corresponding $w_j$'s.

For given $\ell \in \{0,\dots,L\}$, if no undesirable split happens during $\Split(A^x(\ell'-1,i), \ell')$ and $\Split(A^y(\ell'-1,i), \ell')$, for any $\ell' < \ell$ and $i$,
then for each $\ell' < \ell$, the number of blocks $B^x(\ell',i)$ and $B^y(\ell',i)$ will be the same, i.e., $s^x_{\ell'} = s^y_{\ell'}$, and blocks 
$B^x(\ell',i)$ and $B^y(\ell',i)$ will have matching canonical decomposition. 
The total number of $u_j$'s in canonical decomposition of all $B^x(\ell',i)$, $i=1,\dots,t^x_{\ell'}$, will be at most $k$, and similarly for $v_j$'s.
Thus, there will be at most $(4R+24+2)k+2$ positions where an undesirable split can happen in $\Split(A^x(\ell-1,i), \ell)$
for any $i$. Similarly, there are at most $(4R+26)k+2$ positions where an undesirable split can happen in $\Split(A^y(\ell-1,i), \ell)$.
By union bound, the probability that an undesirable split happens in some $\Split(A^y(\ell-1,i), \ell)$ or $\Split(A^y(\ell-1,i), \ell)$, 
for some $\ell$ and $i$, is at most $2(4R+28)k(L+1)/D \le 11 R k (L+1) / D \le 1/10$.

Thus, if no undesirable split happens there are at most $k$ indices $i$ for which the canonical decomposition of $B^x(\ell,i)$ contains some $u_j$.
All other blocks $B^x(\ell,i)$ have a canonical decomposition consisting of a single block $w_0$, for various $w_0$ depending on $\ell$ and $i$. 
Similarly, the canonical decomposition of $B^y(\ell,i)$ contains $v_j$ if and only if $B^x(\ell,i)$ contains $u_j$. 
Blocks $B^y(\ell,i)$ that do not contain $v_j$ are identical to $B^x(\ell,i)$ so they have the same grammar.

Hence, if no undesirable split happens, item 3 of Theorem~\ref{t-decomposition} will be satisfied.

\smallskip
The following theorem generalizes item 3 of Theorem~\ref{t-decomposition} and it will be useful to construct the rolling sketch in Section~\ref{s-rolling}. 

\begin{theorem}\label{t-decomposition-rolling}
    Let $u,v,x,y\in \Sigma^*$ be strings such that $|ux|,|vy|\le n$ and $\ED(x,y) \le k$. Let $G^x_1,\dots,G^x_s$ and  $G^y_1,\dots,G^y_{s'}$ 
    be the sequence of grammars output by the decomposition algorithm on input $ux$ and $vy$ respectively, using the same choice of random functions
    $C_1,\dots,C_L$ and $H_0,\dots,H_L$. With probability at least $1-1/5$ the following is true: There exist integers $r,r',t,t'$ such that $s-t=s'-t'$,
        $$x = \eval(G^x_t)[r,\dots] \cdot \eval(G^x_{t+1}) \cdots \eval(G^x_s) \;\;\;\&\;\;\;y=\eval(G^y_{t'})[r',\dots] \cdot \eval(G^y_{t'+1})\cdots \eval(G^y_{s'}),$$
     and $$\ED(x,y)=\ED(\eval(G^x_{t})[r,\dots], \eval(G^y_{t'})[r',\dots]) + \sum_{i>0} \ED(\eval(G^x_{t+i}),\eval(G^y_{t'+i})).$$
\end{theorem}

Its proof is a minor modification of the proof above. 
We start with the canonical decomposition of $x=w_0 u_1 w_1 \cdots u_k w_k$ and $y=w_0 v_1 w_1 \cdots v_k w_k$, 
form the decomposition $ux=uu_0w_0 u_1 w_1 \cdots u_k w_k$ and $vy=vv_0 w_0 v_1 w_1 \cdots v_k w_k$ where $u_0=v_0=\eps$, 
and follow the compression and split procedures. 
We want to argue that during each split operation, all splits occur either in $w_j$'s and are the same on $ux$ and $vy$, or they occur in $u$ or $v$ where we do not care for them.
Again we define a split to be {\em undesirable} if it starts a new block at a position $j$ that belongs to one of the 
$u_0,u_1,\dots,u_{k}$, $v_0,v_1,\dots,v_{k}$ or it is the position of the first or last symbol of $w_0,w_1,\dots$ or $w_{k}$.
Inductively we maintain that whenever a block $B^{ux}(\ell,i)$ contains a descendant of the compression of $u_j$, 
its corresponding block $B^{vy}(\ell,i')$ contains a descendant of the compression of $v_j$. 
(Here, the correspondence is counting from the highest index $i$ to the lowest and similarly for $i'$, 
so $B^{ux}(\ell,i)$ corresponds to $B^{vy}(\ell,i')$ if $i-i'=s^{ux}_\ell- s^{vy}_\ell$.)
If the blocks contain a descendant of $u_0$ and $v_0$, resp., then we apply Lemma~\ref{l-compress-rolling} to construct a descendant decomposition after their compression. 
For all other blocks that contain some $w_j, u_j$ or $v_j$ we use Lemma~\ref{l-compress} to construct its descendant decomposition. 
We do not care for decomposition of blocks $B^{ux}(\ell,i)$ that are descendants of $u$ but do not contain $u_0$, 
and similarly we do not care for decomposition of blocks $B^{vy}(\ell,i)$ that are descendants of $v$ but do not contain $v_0$.
(They might be decomposed arbitrarily so the number of blocks that are descendants of $u$ might differ from the number of blocks that 
are descendants of $v$.)
Inductively, there are at most $2(4R+28)(k+1)$ positions where an undesirable split can happen in blocks $B^{ux}(\ell,i)$ and $B^{vy}(\ell,i)$
for given level $\ell$. 
In total there are at most $2(4R+28)(k+1)(L+1)$ positions where an undesirable split can happen.
Thus, the probability of making an undesirable split during a run of the algorithm is bounded 
by $2(4R+28)(k+1)(L+1)/D \le 22 R k (L+1) / D \le 1/5$. 
If no undesirable split ever happens then the symbols that are weak compression of symbols from $x$ and $y$ are contained withing the corresponding blocks $B^{ux}(\ell,i)$ and $B^{yv}(\ell,i')$. For the blocks $B^{ux}(\ell,i)$ and $B^{vy}(\ell,i')$ that contain descendants of $u_0$ and $v_0$
it is fine if their prefixes that descend from $u$ and $v$, resp., which are to the left of the descendants of $u_0$ and $v_0$, are split differently
in $B^{ux}(\ell,i)$ and $B^{vy}(\ell,i')$. This does not affect the correspondence between blocks $B^{ux}(\ell,i)$ and $B^{vy}(\ell,i')$ that
weakly come from $x$ and $y$.
This concludes the proof of Theorem~\ref{t-decomposition-rolling}.

\subsection{Encoding a grammar}\label{s-binencoding}

We will set a parameter $N \ge n^{3}$ to be a suitable integer: Let $\FKR : \{0,1\}^* \rightarrow \{1,\dots,N\}$ be 
a hash function picked at random, such as Karp-Rabin fingerprint \cite{rabin_karp},
so for any two strings $u,v \in \{0,1\}^*$, if $u\neq v$ then $\Pr_{\FKR}[\FKR(u)=\FKR(v)] \le (|u|+|v|)/N$.

Set $M= 3 S \cdot \lceil 1 + \log |\Gamma| \rceil$. 
We will encode a grammar $G$ over $\Gamma$ of length at most $S$ given by our decomposition algorithm by a string $\Enc(G)$ over alphabet $\{1,\dots,2N\}$ of length $M$. 
The encoding is obtained as follows: 
First, order the rules of the grammar $G$ lexicographically. 
Then encode the rules in binary one by one using  $3 \cdot \lceil 1 + \log |\Gamma| \rceil$ bits for each rule. (The extra bit allows to mark unused symbols.)
This gives a binary string of length at most $M$, which we pad by zeros to the length precisely $M$. 
We call the resulting binary string $\Bin(G)$.
Compute $h_G = \FKR(\Bin(G))$.
We replace each 0 in $\Bin(G)$ by $h_G$, and each 1 in $\Bin(G)$ by $N+h_G$ to obtain the string $\Enc(G)$.
Clearly, $\Enc(G)$ is a string over alphabet $\{1,\dots,2N\}$ of length exactly $M$. 
The encoding can be computed in time $O(M)$.
For completeness, we encode any grammar $G$ of length more than $S$ or that uses rules with more than two symbols on the right as $\Enc(G)=1^M$.

By the property of $\FKR$ the following holds.

\begin{lemma}\label{l-binenc}
    Let $G,G'$ be two grammars of size at most $S$ output by our decomposition algorithm. Let $\FKR$ be chosen at random.
    \begin{enumerate}
        \item $\Enc(G) \in \{1,\dots, 2N\}^M$.
        \item If $G=G'$ then $\Enc(G)=\Enc(G')$.
        \item If $G\neq G'$ then with probability at least $1-(2M/\alpha)$, $\Ham(\Enc(G), \Enc(G'))=M$, that is the encodings differ in every symbol.
    \end{enumerate}
\end{lemma}

\subsection{Edit distance sketch}\label{s-edsketch}

Let $n$ and $k\le n$ be two parameters, and $p\ge 2N+1$ be a prime such that $p\ge (nM)^3$.
For a string $x \in \Sigma^*$ of length at most $n$, we compute its sketch by running first the decomposition algorithm of Theorem~\ref{t-decomposition}
to get grammars $G_1,G_2,\dots, G_s$. Encode each grammar $G_i$ by encoding $\Enc(G_i)$ from Section~\ref{s-binencoding} using the same $\FKR$ picked at random.
Concatenate the encoding to get a string $w = \Enc(G_1)\cdot \Enc(G_2) \cdots \Enc(G_s)$. 
Calculate the Hamming sketch $\SKHAM_{n',m',p}(w)$ on $w$ for strings of length $n'=nM$ and Hamming distance at most $k'=kM$ from Section~\ref{s-hammin}.
Set the sketch $\SKED_{n,k}(x) = \SKHAM_{n',k',p}(w)$.
The calculation of $\SKED_{n,k}(x)$ can be done in time $\OO(nk)$ as the number of grammars is at most $n$ and each grammar requires $\OO(k)$ time to be encoded into binary. The Hamming sketch can be constructed in time $\OO(nk)$.
(We believe that on average we expect only $\OO(n/k)$ grammars to be produced for a given string $x$ so the actual running time should be $\OO(n)$ on average.)

\begin{theorem}
Let $x,y \in \Sigma^*$ be strings of length at most $n$ such that $\ED(x,y) \le k$. 
Let $\SKED_{n,k}(x)$ and $\SKED_{n,k}(y)$ be obtained using the same randomness for the decomposition algorithm and the same choice of $\FKR$.
With probability at least $2/3$, we can calculate $\ED(x,y)$ from $\SKED_{n,k}(x)$ and $\SKED_{n,k}(y)$.
\end{theorem}

Assume that the output of the decomposition algorithm on $x$ and $y$ satisfies all the conclusions of Theorem~\ref{t-decomposition}.
In particular, for $x$ we get $\eval(G^x_1)\cdot \eval(G^x_2) \cdots \eval(G^x_s)$ 
and for $y$ we get $\eval(G^y_1) \cdots \eval(G^y_s)$, for some $s\le n$,
each of the grammars is of size at most $S$, $\ED(x,y)=\sum_i \ED(\eval(G^x_i),\eval(G^y_i))$,
and the number of pairs $G^x_i$ and $G^y_i$ where $G^x_i\neq G^y_i$ is at most $k$.
Assume that $\FKR$ is chosen so that $\Enc(G^x_i) \neq \Enc(G^y_i)$ for each of the pairs where  $G^x_i$ and $G^y_i$ differ. 

In order to determine $\ED(x,y)$, we recover the (Hamming) mismatch information between $\Enc(G^x_1)\cdot \Enc(G^x_2) \cdots \Enc(G^x_s)$
and $\Enc(G^y_1)\cdot \Enc(G^y_2) \cdots \Enc(G^y_s)$ from $\SKED_{n,k}(x)$ and $\SKED_{n,k}(y)$.
That gives grammars $G^x_i$ and $G^y_i$, for all $i$ where $G^x_i\neq G^y_i$. 
(Whenever the two grammars differ, their encoding differ in every symbol by Lemma~\ref{l-binenc} so we can recover them from the Hamming mismatch information.)
Calculating the edit distance of each of the pair of differing grammars using the algorithm from Proposition~\ref{p-edgrammar} we recover $\ED(x,y)$ as the sum of their edit distances.

The sum is correct unless some of the assumptions fail:
The probability that the grammar decomposition fails (does not have properties from Theorem~\ref{t-decomposition}) for the pair $x$ and $y$ is at most $1/5$ for $n$ large enough. 
The probability that the choice of $\FKR$ fails (two distinct grammars have the same encoding) is at most $2kM/N < 1/n$ by the choice of $N$. 
The probability that the Hamming distance sketch fails to recover the mismatch information between all the grammars is at most $1/n$.
So in total, the probability that the output of the algorithm is incorrect is at most $1/3$.

The running time of the comparison algorithm is $\OO(k^2)$: The Hamming mismatch information can be recovered in time $\OO(kM)=\OO(k^2)$ (Proposition~\ref{p-skham}), then we build the $\le k$ mismatched grammars in time $\OO(k^2)$, and run the edit distance computation
on the pairs of grammars in time $\sum_{i<k} \OO(k+k_i^2) \le \OO(k^2)$, where $k_i$ is the edit distance of the $i$-th pair of mismatched grammars.
(We interrupt the edit distance computation if it takes more time than $\OO(k^2)$ which would indicate $\ED(x,y)>k$.)

To decide whether $\ED(x,y) > k$ we note that on input $x$ and $y$, the Hamming sketch either outputs the correct mismatched places if their number
is $\le k'$ or it outputs $\infty$ if there are more mismatches than that or the sequences sketched by the Hamming sketch are of different length. 
(We assume that the Hamming sketch knows the number of symbols it is sketching.)
In the $\infty$-case we know that there are more than $k$ different pairs of grammars or the decomposition of $x$ and $y$ failed, and we can report  $\ED(x,y) > k$. In the other case we try to calculate the edit distance of the differing pairs of grammars. 
If we spend more than $\OO(k^2)$ time on it or we get a number larger than $k$ then we report $\ED(x,y)>k$.
This correctly decides whether $\ED(x,y)>k$ with probability at least $2/3$.

To prove Theorem~\ref{t-main2} we build a more robust sketch by taking $c\log n$ independent copies of the sketch $\SKED_{n,k}$.
To calculate the edit distance of two sketched strings we run the edit distance calculation on each of the corresponding pairs of copies,
and output the majority answer. A usual application of Chernoff bound shows that the probability of correct answer is at least $1-1/n$ for suitable
constant $c>0$.

\section{Rolling sketch for edit distance}\label{s-rolling}

In this section we will construct the rolling sketch of Theorem~\ref{t-main3}. We will use two claims that will be proved in Section~\ref{s-rollingproofs}. The first one addresses how much a compression of a string $w$ might change depending on what is appended to it.

\begin{lemma}\label{l-extension}
Let $\ell\in \{0,\dots,L\}$ and $v,u,w\in \Gamma^*$. Let $w'u'=\Compress(wu,\ell)$ and let $w''v'=\Compress(wv,\ell)$, where $w'$ is the compression
of $w$ when compressing $wu$ and $w''$ is the compression of $w$ when compressing $wv$. Let $t=|w'|-3(R+1)$ or $t=|w'u'|-|u|-3(R+1)$. Then $w'[1,t]=w''[1,t]$.
% note: a tight  bound is probably something like R/3 instead of 2R.
\end{lemma}

The next lemma addresses how much the overall decomposition of a string $x$ might change if we append a suffix $z$ to it.

\begin{lemma}\label{l-grammarsuffix}
Let $x,z\in \Sigma^*$, $|xz|\le n$. Let $H_0,\dots, H_L, C_1,\dots, C_L$ be given. 
Let $G^x_1,G^x_2,\dots,G^x_s$ be the output of the decomposition algorithm on input $x$, and
 $G^{xz}_1,G^{xz}_2,\dots,G^{xz}_{s'}$ be the output of the decomposition algorithm on input $xz$ using the given hash functions.
Let $T=L(3R+6)$. 
\begin{enumerate}
    \item $G^x_i=G^{xz}_i$ for all $i=1\dots,s-T$.
    \item $|x| \le \sum_{i=1}^{\min(s+T,s')} |\eval(G^{xz}_i)| $.
\end{enumerate}
\end{lemma}

The second part says that if $x$ is decomposed into $s$ grammars by itself, then it can be recovered from the first $s+T$ grammars for $xz$.
Hence, appending extra symbols to $x$ cannot increase the number of grammars that cover $x$ by more than $T$.

Let $m\ge k$ and $n\ge 10m^3$ be integers. A rolling sketch for a string obtained by up-to $m$ insertions (to the right end) and $m$ deletions (from the left end)
from an empty word consists of three data structures: {\em insertion buffer}, {\em deletion buffer} and a Hamming distance sketch $\SKHAM_{n',k',p}$, where $k'=(4T+1)(k+2)M$, $n'=nM$ and $p\ge n'^3$ is a chosen prime.

The insertion buffer maintains a buffer of {\em committed grammars} $G_{s-4T+1},G_{s-4T+2},\dots, G_{s}$ and a buffer of {\em active grammars} $G^i_1,\dots,G^i_t$,
$t \le T$. The deletion buffer is similar, it maintains a buffer of {\em committed grammars} $G_{r-4T+1},G_{r-4T+2},\dots, G_{r}$ and a buffer of {\em active grammars} $G^d_1,\dots,G^d_{t'}$, $t' \le T$. The Hamming sketch is a sketch of grammars $G_{r-2T+1},G_{r-2T+2},\dots, G_{s-2T}$, each encoded
as a string of length $M$ over the alphabet $\{1,\dots,2N\}$. 

In addition to that, the sketch keeps track of the current
value of $r$ and $s$, and remembers a collection of pair-wise independent hash functions $C_1,\dots, C_L$, a collection of $(5D \log n)$-wise independent hash functions $H_0,\dots, H_L$, and randomness for Karp-Rabin fingerprint to compute binary encoding of grammars.
The hash functions and the randomness of Karp-Rabin fingerprint are chosen at random when creating the sketch for empty string.
This extra information requires $\OO(k)$ bits to specify.

Initially, the committed grammars in the insertion and deletion buffers are all treated as empty sets, there are no active grammars in the insertion or deletion buffers so $t=t'=0$ and $s=r=0$.

For $u,x \in \Sigma^*$, if in total a string $ux$ was inserted into the sketch 
then $G_{1},\dots, G_{s}, G^i_1,\dots,G^i_t$ represents $ux$,
that is $ux$ is the concatenation of the evaluation of the grammars.
If in total the string $u$ was deleted from the sketch, then $G_{1},\dots, G_{r}, G^d_1,\dots,G^d_{t^d}$ represents $u$. (See Fig.~\ref{fig:rolling_sketch} for an illustration.)

\begin{figure}[htp]
    \centering
    \includegraphics[width=\textwidth]{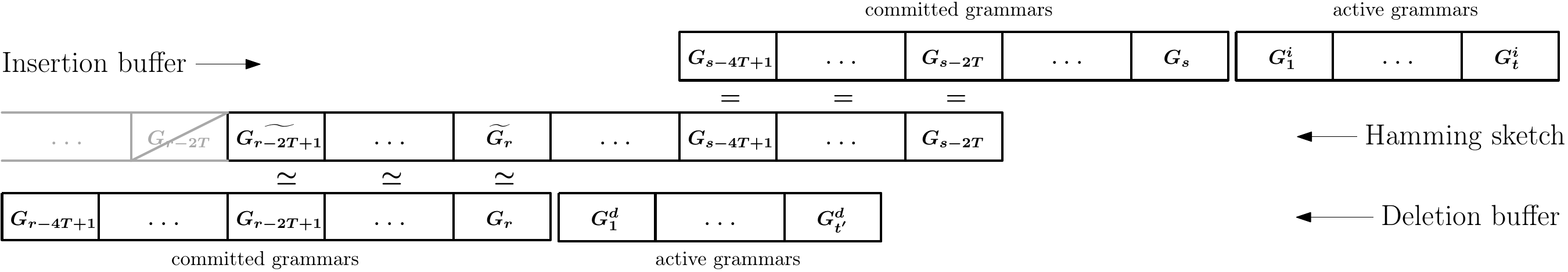}
    \caption{Rolling sketch.}
    \label{fig:rolling_sketch}
\end{figure}

\smallskip\noindent{\em Appending a symbol.}
When we append additional symbol $a$ to the sketch we modify input buffers as follows: 
We update the active grammars $G^i_1,\dots,G^i_t$ by appending $a$ as explained further below. Say the update produces grammars $G'^i_1,\dots,G'^i_{t'}$.
If $t'\le T$ then the produced grammars will become the active grammars, and no more changes are done to the sketch. 
Otherwise we commit the first $t'-T$ grammars $G'^i_1,\dots,G'^i_{t'-T}$ one-by-one into the committed buffer as grammars $G_{s+1},\cdots,G_{s+t'-T}$ and we keep the remaining grammars as the active grammars. 

Committing a grammar $G_{s+1}$ into the committed buffer 
will trigger addition of $G_{s-2T+1}$ into the Hamming sketch at the end of the represented sequence of grammars
(if $s-2T+1>0$), and removing the grammar $G_{s-4T+1}$
from the committed buffer. For insertion into the Hamming sketch, the grammar $G_{s-2T+1}$ is encoded into binary as in Section~\ref{s-binencoding}
and then the binary string is encoded using the Karp-Rabin fingerprint $\FKR$ of {\em all} the grammars $G_{s-4T+1},\dots, G_{s+1}$, instead of only the grammar $G_{s-2T+1}$. (Thus, a change in any of the neighboring grammars will trigger a recovery of also the grammar $G_{s-2T+1}$ when calculating a mismatch information
from the Hamming sketch.) We repeat this process for each grammar being committed.

By the second part of Lemma~\ref{l-grammarsuffix} $t'\le t+T \le 2T$ so we will commit at most $T=\OO(1)$ grammars. 
It takes time $O(MT)=\OO(k)$ to prepare the binary encoding of each of the committed grammars, and $\OO(k^2)$ to insert it into the Hamming sketch. 
The update of the active grammars takes $\OO(k)$ time as described below. So in total this step takes $\OO(k^2)$ time.

\smallskip\noindent{\em Removing a symbol.}
Deletion buffer works in manner similar to insertion buffer, we add the removed symbol $a$ to the active grammars, but when committing the grammar $G_{r+1}$, we use $\FKR$-fingerprint of all the grammars $G_{r-4T+1},\dots, G_{r+1}$
to encode grammar $G_{r-2T+1}$ which is then {\em removed} from the beginning of the sequence of grammars represented by the Hamming sketch (if $r-2T+1>0$), 
i.e., we update the Hamming sketch to reflect this removal. Similarly to appending a symbol, this step takes time $\OO(k^2)$.

\smallskip\noindent{\em Active grammar update.}
The update of active grammars $G^i_1,\dots,G^i_t$ when appending $a$ is done as follows. $G_{1},\dots, G_{s}, G^i_1,\dots,G^i_t$ represents $ux$
so we need to calculate the grammars for $uxa$. We claim that only the active grammars might change: At some point, $G_{s}$ became committed so at that time there was $T$ active grammars following it. If at that point the grammars together represented a string $z$, by appending more symbols to $z$ we cannot change grammars 
$G_{1},G_{1},\dots, G_{s}$ according to the first part of Lemma~\ref{l-grammarsuffix}. So appending $a$ to $ux$ will affect only the active grammars.

From the analysis in the proof of Lemma~\ref{l-grammarsuffix} it follows that for $\ell \in \{0,\dots,1\}$ 
if $B^{ux}(\ell,1),\dots,B^{ux}(\ell,s^{xy}_\ell)$ is the trace of the decomposition algorithm on $ux$ at level $\ell$, 
and $B^{uxa}(\ell,1),\dots,B^{uxa}(\ell,s^{xya}_\ell)$ is the trace on $uxa$, then their difference spans at most $\ell (3R+6)$ last symbols
of $B^{ux}(\ell,1) \cdots B^{ux}(\ell,s^{xy}_\ell)$.

So instead of decompressing the active grammars completely, adding $a$ and recompressing them back, 
we only decompress the necessary part of each trace $B^{ux}(\ell,1) \cdots B^{ux}(\ell,s^{xy}_\ell)$. 
Let $\# \rightarrow v_i$ be the starting rule of the active grammar $G_i$. 
Starting from the string $v_1 \cdot v_2 \cdots v_t$, for each $\ell=L,\dots,1$, 
we iteratively rewrite all level-$\ell$ symbols in the string using the appropriate grammars 
while only maintaining at most $T$ last symbols of the resulting string. 
(Care has to be taken to maintain information about any sequence $a^r$ stretching from those $T$ last symbols to the left.)

We add $a$ to the resulting string and re-apply compress and split procedures for levels $0,1,\dots, \ell-1$ to 
recompress only the part of the trace affected by modifications.
As we perform the compression of symbols we maintain a set $G$ of all grammar rules needed for decompression.
(We initialize $G$ with the union of all rules from the active grammars $G^i_1,\dots,G^i_t$ minus the starting rules, 
and we iteratively add new rules coming from the recompression.)
For the recompression we need to know the context of up-to $R+1$ symbols preceding the modified part of the trace.
On the other hand, the modification can affect the recompression of up-to $R+1$ symbols to the left from the left-most modified symbol in the trace.
Those $R+1$ symbols all happen to be within the decompressed suffix of the trace of size at most $T$.

Eventually, we get a new level-$L$ trace $B^{uxa}(L,s^{xya}_L-t'+1),\dots,B^{uxa}(L,s^{xya}_L)$, for some $t'$. 
Each new grammar $G'^i_j$ is obtained by taking the grammar $G \cup \{ \# \rightarrow B^{uxa}(L,s^{xya}_L-t'+j) \}$ and removing from it all useless rules. This can be done in time $O(|G|)$. (See Section~\ref{s-grammars}).

Overall the update of active grammars on insertion of a single symbol will require $O(LT)=\OO(1)$ evaluations of split hash functions $H_0,\dots,H_L$, 
$O(LT)=\OO(1)$ evaluations of compress hash functions $C_1,\dots,C_L$, and $O(T (LT + \sum_{j=1}^t |G^i_{j}|))$ time to produce the new grammars. 
As the total size of the grammars is $\OO(k)$ and the time to evaluate $H_\ell$ at a single point is $\OO(1)$, the overall time for the update of active grammars is $\OO(k)$. We provide a more detailed description of the update procedure in Section~\ref{s-update}.

\smallskip\noindent{\em Edit distance evaluation.}
Consider strings $x$ and $y$ of length at most $m$ and edit distance at most $k$. Consider the rolling sketch $\SKR_{m,k}(x)$ for $x$ obtained by
inserting symbols $ux$ and removing symbols $u$, for some $u\in \Sigma^*$ where $|ux|\le m$.
Consider also the rolling sketch for $y$ obtained by
inserting symbols $vy$ and removing symbols $v$, for some $v\in \Sigma^*$ where $|vy|\le m$. 
Both sketches should use the same randomness that is to start from the same sketch for empty string.

The rolling sketch for $x$ consists of the insertion buffer with committed grammars  $G^x_{s^x-4T+1},$ $G^x_{s^x-4T+2},$ $\dots,$ $G^x_{s^x}$ and with active grammars $G^{ix}_1,\dots,G^{ix}_{t^x}$, and the deletion buffer with committed grammars $G^x_{r^x-4T+1},$ $G^x_{r^x-4T+2},$ $\dots,$ $G^x_{r^x}$ and active grammars $G^{dx}_1,\dots,G^{dx}_{t'^x}$, $t'^x \le T$. Its Hamming sketch sketches the sequence of grammars $G^x_{r^x-2T+1},G^x_{r^x-2T+2},\dots, G^x_{s^x-2T}$. Similarly for $y$, we have the committed insertion grammars $G^y_{s^y-4T+1},G^y_{s^y-4T+2},\dots, G^y_{s^y}$, etc.

We extend the notation so for $j\in \{1,\dots,t^x\}$, we let $G^x_{s^x+j}$ denote the active grammar $G^{ix}_j$, and similarly for $y$.
Let $d^x = s^x + t^x - r^x$ and $d^y = s^y + t^y - r^y$.
We assume that the hash functions used to decompose $ux$ and $vy$ into grammars satisfy the probabilistic conclusion of Theorem~\ref{t-decomposition-rolling}. That means that grammars $G^x_{r},\dots$ and $G^y_{r},\dots$ can be aligned from the right
so $G^x_j$ corresponds to $G^y_{j-d^x + d^y}$, for $j\ge r^x$ (they might not be identical because of the edit operations). 
Without loss of generality we assume that $d^x \ge d^y$.

Before proceeding with the algorithm we first observe that $d^x-d^y < 2T$.
Let $p^x\ge r^x+1$ be the index of the grammar $G^x_{p^x}$ which produces the first symbol of $x$ when we evaluate all the grammars.
Similarly, $p^y\ge r^y+1$ is the index of $G^y_{p^y}$ which produces the first symbol of $y$. By Lemma~\ref{l-grammarsuffix}
applied on $x\leftarrow u$ and $z \leftarrow x$ we get that $p^x \le r^x + t'^x + T \le r^x + 2T$, and similarly $p^y \le r^y + 2T$.
By our assumption on success of Theorem~\ref{t-decomposition-rolling}, $s^x + t^x - p^x = s^y - t^y - p^y$.
Hence, $s^x + t^x - s^y - t^y = p^x - p^y \le r^x + 2T - r^y - 1 \le r^x - r^y + (2T-1)$.
Thus $d^x - d^y = s^x + t^x -r^x - s^y - t^y + r^y \le r^x - r^y + (2T-1) - r^x + r^y \le 2T-1$.

If $d^x < 10T$ then we can recover all the grammars $G^x_{r^x-2T+1},G^x_{r^x-2T+2},\dots, G^x_{s^x-2T}$ from their Hamming sketch by constructing an auxiliary {\em dummy} Hamming sketch $sk'$ for a sequence of $1$'s of length $(s^x-r^x)M$
and comparing the two sketches. ($M$ is the length of the encoding of each grammar.)
Their mismatch information reveals all the grammars $G^x_{r^x-2T+1},\dots, G^x_{s^x-2T}$
Since $d^y \le d^x$, we can similarly recover all the grammars $G^y_{r^y-2T+1},\dots, G^y_{s^y-2T}$ from their Hamming sketch.

Thus we know all grammars $G^x_{r^x+1},G^x_{r^x+2},\dots, G^x_{s^x+t^x}$ and $G^x_{r^y+1},G^y_{r^y+2},\dots, G^y_{s^y+t^y}$.
We also know grammars  $G^{dx}_1,\dots,G^{dx}_{t'^x}$ and  $G^{dy}_1,\dots,G^{dy}_{t'^y}$ that need to be {\em subtracted}
from our grammars. As noted in Section~\ref{s-grammars}, for each of the grammars we can calculate its evaluation size.
From that information we can easily identify $p^x$ and $p^y$, and shorten the grammars $G^x_{p^x}$ and $G^y_{p^y}$ to produce
only symbols of $x$ and $y$, respectively. 
We can combine all the grammars of $x$ into one grammar $G^x$, and all the grammars of $y$ into $G^y$, 
and run the algorithm of Ganesh, Kociumaka, Lincoln and Saha~\cite{ED_compressed_string_Soda22} to calculate the edit distance of $x$ and $y$. 
Since $T=\OO(1)$, that will take time  $\OO(|G^x|+|G^y| + k^2) = \OO(k^2)$.

If $d^x \ge 10T$ then we proceed as follows. Clearly, $d^y \ge 8T$, so $s^y - r^y \ge 7T$  and $s^x - r^x  \ge 9T$. 
Thus  $G^x_{r^x-2T+1},G^x_{r^x-2T+2},\dots, G^x_{s^x-2T}$ and  $G^y_{r^y-2T+1},G^y_{r^y-2T+2},\dots, G^y_{s^y-2T}$
consist of at least $7T$ grammars each, and those grammars are sketched by their Hamming sketches.
Although we assume that there is a correspondence between the grammar $G^x_j$, for $j\ge r^x$, and $G^y_{j-d^x+d^y}$
the sequences $G^x_{r^x-2T+1},\dots, G^x_{s^x-2T}$ and  $G^y_{r^y-2T+1},\dots, G^y_{s^y-2T}$ are misaligned in their Hamming sketches by $d^x-d^y$ grammars. 
To rectify this misalignment, we prepend $(d^x-d^y)M$ copies of symbol $1$ into the sketch for $G^y_{r^y-2T+1},\dots, G^y_{s^y-2T}$.
Furthermore, if $t^x < t^y$ then we append $(t^y - t^x)M$ ones into the sketch for $G^y_{r^y-2T+1},\dots, G^y_{s^y-2T}$,
to rectify the difference in the number of sketched grammars. 
Otherwise if $t^x > t^y$ then we append $(t^x - t^y)M$ ones into the sketch for $G^x_{r^x-2T+1},\dots, G^x_{s^x-2T}$.

Now we can calculate the mismatch information from the Hamming sketches to find out the pairs of grammars
$G^x_j$ and $G^y_{j-d^x+d^y}$, $j\ge r^x + 1$, that are different. 

If for some $j\in \{r^x + 1,\dots, r^x+2T\}$, $G^x_j$ and $G^y_{j-d^x+d^y}$ differ then because we use the Karp-Rabin
fingerprint of the two grammars to encode also the neighboring grammars up-to distance $2T$, we recover from the sketch
all the grammars $G^x_j$ and $G^y_{j-d^x+d^y}$, for $j=r^x + 1,\dots, r^x+2T$. 
By counting the evaluation size of each of those grammars and comparing it with the evaluation size of active grammars in deletion buffers of $x$ and $y$, resp., 
we identify $p^x$ and $p^y$, and how much the grammars $G^x_{p^x}$ and $G^y_{p^y}$ should be shortened to produce
only symbols of $x$ and $y$. 
After shortening $G^x_{p^x}$ and $G^y_{p^y}$ we calculate the edit distance of their evaluation.
We sum it up with the edit distance of evaluation of each pair of grammars $G^x_j$ and $G^y_{j-d^x+d^y}$, for $j>p^x$, 
that was identified as mismatch by the Hamming distance sketch or that belongs among the active grammars in insertion buffers of either $x$ or $y$. 
There will be at most $T$ mismatched pairs involving the active grammars, and $(4T+1)k$ pairs identified by the Hamming sketch. 

In the remaining case when $G^x_j$ and $G^y_{j-d^x+d^y}$ are identical for all $j\in \{r^x + 1,\dots, r^x+2T\}$,
we might not be able to recover all those grammars from the Hamming sketches,
and we might not be able to identify $p^x$ and $p^y$. 
However, since $G^x_{p^x}=G^y_{p^y}$, we know that the part of $x$ produced by $G^x_{p^x}$ is either a prefix or 
suffix of the part of $y$ produced by $G^y_{p^y}$. 
The difference in the size of the two parts is the edit distance of the two parts.
The difference is given by the difference between the total evaluation size of active grammars in the deletion buffer of $x$,
and the total evaluation size of active grammars in the deletion buffer of $y$ together with grammars $G^y_{r_y-j}$,
for $j=0,\dots, d^x-d^y-1$. 
The latter grammars are in the committed deletion buffer of $y$ and they agree with $G^x_{r^x+1},\dots, G^x_{r^x+d^x-d^y}$.
Hence, the edit distance of the parts of $x$ and $y$ coming from $G^x_{p^x}$ and $G^y_{p^y}$ can be determined.
All other mismatching pairs of grammars are identified by the Hamming sketch or are among active grammars of the insertion buffers. 
So we proceed as in the previous case to calculate their contribution to the edit distance of $x$ and $y$.
The edit distance of $x$ and $y$ is the sum of those edit distances.

We see that in both the cases we need the Hamming sketch to be able to recover at least $T$ mismatched grammars at the very end caused by
the dummy padding, $4T$ grammars at the beginning corresponding to $G^x_{r^x-2T+1},G^x_{r^x-2T+2},\dots, G^x_{r^x+2T}$, 
$2T$ neighbors of $G^x_{r^x+2T}$ to the right, 
and at most $(4T+1)k$ mismatched grammars caused by the edit operations between $x$ and $y$.
This is less than $M(4T+1)(k+2)$ which is the number of mismatches our Hamming sketch can recover.

The time needed to compare the sketched strings can be bounded as follows: 
In total the procedure generates at most $O(T k)$ pairs of grammars of total size $\OO(k^2)$ on which it runs edit distance computation from Proposition~\ref{p-edgrammar}.
If those edit distance computations take total time more than $\OO(k^2)$ we can terminate them as we know the overall edit distance is larger than $k$.
Recovering differing grammars from the Hamming distance sketch takes time $\OO(k')=\OO(k^2)$. 
Their follow-up processing such as counting their evaluation size and shortening them is proportional to their total size which is $\OO(k^2)$.
Hence, the time for comparing strings is $\OO(k^2)$.

\bigskip\noindent{\em Failure probability.}
The update operations can fail if the grammar decomposition produces large grammars or the grammars are not deterministic (because of a collision caused by compression hash functions). 
This happens with probability at most $2/n$ for each update. 
Since we perform at most $m$ updates, the failure probability of any update operations is at most $2m/n \le 1/5m^2$ by our choice of $m$ and $n$.

When comparing two sketches for strings $x$ and $y$ of edit distance at most $k$, 
Theorem~\ref{t-decomposition-rolling} might fail to align their grammar decomposition.
This happens with probability at most $1/5$. 
With probability at most $2/n$ the Hamming sketches might fail to recover the differing pairs of grammars.
There is no other source of failure for strings of edit distance at most $k$ so the probability of the compare operation failing is at most $1/3$.
To boost the success probability of comparison from $2/3$ to $1-1/2m$, 
we again form a more robust sketch by taking $c\log m$ independent copies of the rolling edit distance sketch
and operate on them simultaneously. For comparison we output the most frequent answer from the individual sketches.
This multiplies the failure probability of update operations by $c\log m$, so it is still at most $1/2m$ for $m$ large enough.
The comparison will fail with probability at most $1/2m$.

For strings of edit distance more than $k$ the comparison of an individual edit sketch will fail either 
because the Hamming sketch would need to recover more than $k$ pairs of differing grammars or 
because the total edit distance of the differing grammars is more than $k$. 
In both these failure cases we can always output $\infty$ to be on the safe side.

% \subsection{Approximate pattern matching with $k$ edits}

% To design a small space streaming algorithm for the approximate pattern matching with $k$ edits 
% we will proceed as follows. First, we create a rolling sketch for the pattern $P$. Let $m$ be the length of $P$.
% To process the text $T$ we will maintain $2k+1$ rolling sketches indexed by integers from the $I=\{m-k,m-k+1,\dots,m+k\}$. 
% For $i\in I$, at time $t$, the $i$-th sketch will correspond to text $T[t-i+1,t]$

\subsection{Proofs of Lemma~\ref{l-extension} and \ref{l-grammarsuffix}}\label{s-rollingproofs}

Here we prove the remaining two lemmas.

\bigskip
\begin{proofof}{Lemma~\ref{l-extension}}
    For the simplicity of case analysis we first compare the compression of $wu$ and $w$. Consider the division of $w=B_1\dots B_m$
    when calling $\Compress(w,\ell)$, and the division $wu=B'_1\dots B'_{m'}$ when calling $\Compress(wu,\ell)$. Let $w'''=\Compress(w,\ell)$, from line.
    Let $a$ be the last symbol of $w$. We consider three cases. 
    
    If $B_m=a^r$, $r\ge 1$, then $B_i=B'_i$ for all $i=1\dots,m-1$, and $B_m$ is a prefix of $B'_m$, not necessarily proper. In this case,
    the compression of each $B_i$ and $B'_i$, $i=1,\dots, m-1$, is the same, so $w'''$ equals to $w'$ in all but possibly the last symbol.
    
    Otherwise, $B_m$ consists of at least two singleton symbols. If the first symbol of $u$ is $a$, then $B'_m=B_m[1,|B_m|-1]$, and $B'_{m+1}=a^r$, for some $r\ge 2$. $\FL$ will color the same all symbols of $B_m$ and $B'_m$ except for at most the last $R$ symbols. Hence, $B_m$ and $B'_m$ will be compressed the same except for at most the last $R$ symbols. The last $R$ symbols are compressed into at most $R$ symbols in $w'''$, and $B'_{m+1}$ will be compressed into a single symbol. In this case we conclude that $w'[1,|w'|-R-1]=w'''[1,|w'|-R-1]$. 

    If the first symbol of $u$ is not $a$ then $B_m$ is a prefix of $B'_m$, and the compression of $B_m$ and $B'_m[1,|B_m|]$ will differ in at most $R$ last symbols.  So $w'[1,|w'|-R]=w'''[1,|w'|-R]$. 

    Hence, in all three cases  $w'[1,|w'|-R-1]=w'''[1,|w'|-R-1]$. Moreover, $\left| |w'|-|w'''|\right| \le R+1 $. 
    
    A similar argument gives   $w''[1,|w''|-R-1]=w'''[1,|w''|-R-1]$, and $\left| |w''|-|w'''|\right| \le R+1$. By the triangle inequality, $\left| |w'|-|w''|\right| \le 2(R+1)$. Hence, $w'[1,|w'|-3(R+1)]=w''[1,|w'|-3(R+1)]$. Since $|u'|\le |u|$, $|w'u'|-|u| \le |w'|$. The claim follows.
\end{proofof}

\bigskip
\begin{proofof}{Lemma~\ref{l-grammarsuffix}}
    {\em Part 1.}
    Consider strings $B^x(\ell,i)$ from the trace of the algorithm on $x$ given the hash functions $H_0,\dots, H_L, C_1,\dots, C_L$. (See Section~\ref{s-correctness}) Similarly for $B^{xz}(\ell,i)$. 
    
    For $\ell=0,\dots,L$ we will define integers $i_\ell$ and $\Delta_\ell$ satisfying:
    \begin{enumerate}
        \item For all $i<i_\ell$,  $B^x(\ell,i)= B^{xz}(\ell,i)$,
        \item $B^x(\ell,i_\ell)[1,|B^x(\ell,i_\ell)|-\Delta_\ell]= B^{xz}(\ell,i_\ell)[1,|B^x(\ell,i_\ell)|-\Delta_\ell]$,
        \item $\Delta_\ell + \sum_{i=i_\ell+1}^{s^x_\ell} |B^x(\ell,i)| \le \ell (3R+3)+1$.
    \end{enumerate}
    
    For $\ell=0$, $B^x(0,1),B^x(0,2),\dots, B^x(0,s^x_0)=\Split(x,0)$, so we set $i_0=s^x_0$ and $\Delta_0=1$. 
    Since  $B^{xz}(0,1),B^{xz}(0,2),\dots, B^{xz}(0,s^{xz}_0)=\Split(xz,0)$, 
    $s^x_0 \le s^{xz}_0$ and $B^x(0,s^x_0)$ might differ from $B^{xz}(0,s^x_0)$ by containing
    the last symbol of $x$ which might be the first symbol of $B^{xz}(0,s^x_0+1)$. Otherwise, $B^x(0,s^x_0)$ is the prefix of $B^{xz}(0,s^x_0)$
    so the properties of $i_0$ and $\Delta_0$ are satisfied.
    
    For $\ell=1,\dots,L$, having defined $i_{\ell-1}$ and $\Delta_{\ell-1}$ we will define  $i_{\ell}$ and $\Delta_{\ell}$:
    Define 
   $$A^x_{\ell-1}=\Compress(B^x(\ell-1,i_{\ell-1}),\ell) \;\;\;\;\&\;\;\;\;  A^{xz}_{\ell-1}=\Compress(B^{xz}(\ell-1,i_{\ell-1}),\ell),$$       
        $$(B_0,B_1,\dots, B_m)=\Split(A^x_{\ell-1},\ell) \;\;\;\;\&\;\;\;\; (B'_0,B'_1,\dots, B'_{m'})=\Split(A^{xz}_{\ell-1},\ell).$$

    For simplicity of exposition in this proof we assume that for any $B\in \Gamma^*$ of size at most $2$, $\Compress(B,\ell)=B$ and $\Split(B,\ell)=(B)$, so they both perform no action on $B$ of size at most 2.
    
    Let 
    \begin{align*}
        w &= B^x(\ell-1,i_{\ell-1})[1,|B^x(\ell-1,i_{\ell-1})|-\Delta_{\ell-1}],\\ 
        u &= B^x(\ell-1,i_{\ell-1})[1+|B^x(\ell-1,i_{\ell-1})|-\Delta_{\ell-1},\dots],\\
        v &= B^{xz}(\ell-1,i_{\ell-1})[1+|B^x(\ell-1,i_{\ell-1})|-\Delta_{\ell-1},\dots].
    \end{align*} 
    By Lemma~\ref{l-extension}, $A^x_{\ell-1}$ and  $A^{xz}_{\ell-1}$ agree 
    on at least the first $|A^x_{\ell-1}| -(3R+3) - |u|=|A^x_{\ell-1}| -(3R+3) - \Delta_{\ell-1}$ symbols.
    (This is trivial when $|w|\le 2$, in particular, when $|B^x(\ell-1,i_{\ell-1})|\le 2$ or $|B^{xz}(\ell-1,i_{\ell-1})|\le 2$.)

    Let $i\in \{0,\dots,m\}$ be the largest $i$ such that for all $j < i$, $B_j=B'_j$. Let $i_\ell$ be the index of block $B_i$ among
    blocks $B^x(\ell,1), B^x(\ell,2), \dots, B^x(\ell,s^x_\ell)$. 
    Notice, $B^x(\ell,j)= B^{xz}(\ell,j)$, for all $j<i_\ell$.
    Let $\Delta_\ell \ge 0$ be the smallest integer such that $B_i[1,|B_i|-\Delta_\ell]=B'_i[1,|B_i|-\Delta_\ell]$.
    (So the second property holds for $\ell$, as $B^x(\ell,i_\ell)=B_i$ and $B^{xz}(\ell,i_\ell)=B'_i$.)
    Since $B_i[|B_i|-\Delta_\ell+1,\dots] \cdot B_{i+1} \dots B_m$ forms a part of a suffix of $A^x_{\ell-1}$ on which $A^x_{\ell-1}$ and $A^{xz}_{\ell-1}$ differ,
    $\Delta_\ell + \sum_{j=i+1}^{m} |B_j| \le (3R+3)+ \Delta_{\ell-1}$.
    
    Notice, $\sum_{j=i_\ell+1+s-i}^{s^x_\ell} |B^x(\ell,j)| \le \sum_{j=i_{{\ell-1}+1}}^{s^x_{\ell-1}} |B^x(\ell-1,j)|$ as each $B^x(\ell,j)$
    on the left is a part of a compression of some $B^x(\ell-1,j)$ on the right. 
    Hence,  $$\sum_{j=i_{\ell+1}}^{s^x_\ell} |B^x(\ell,j)| + \Delta_\ell \le
             (3R+3) + \Delta_{\ell-1} + \sum_{j=i_{{\ell-1}+1}}^{s^x_{\ell-1}} |B^x(\ell-1,j)| \le (3R+3)+(\ell-1)(3R+3)+1 \le \ell(3R+3)+1
             $$
    
    Eventually, $\Delta_L + \sum_{i=i_L+1}^{s^x_L} |B^x(L,i)| \le L (3R+3)+1$. 
    Also $B^x(L,j)=B^{xz}(L,j)$ for $j=1,\dots,i_L-1$.
    Hence, $G^x_j=G^{xz}_j$ for $j=1,\dots,i_L-1$. 
    Since each $|B^x(\ell,j)|\ge 1$, $s^x_L - i_L \le L(3R+3)+1$, which implies  $s^x_L - (L(3R+3)+2)  \le s^x_L - T  \le i_L-1$ and the claim follows.
    
    \medskip\noindent{\em Part 2.}
    The proof of this part proceeds similarly to the first part. Let $B^{x}(\ell,i)$ and $B^{xz}(\ell,i)$ be as above. For $\ell=0,\dots,L$, we will define a sequence of integers $i_\ell, t_\ell, p_\ell, r_\ell$ satisfying:
    \begin{enumerate}
        \item $t_\ell$ is the last index $i$ such that $B^{xz}(\ell,i)$ contains some symbol that comes from the compression of $x$, and $r_\ell$ is the length of the prefix of $B^{xz}(\ell,t_\ell)$ that comes from the compression of $x$.
        \item $i_\ell \le t_\ell$ and for all $i<i_\ell$,  $B^x(\ell,i)= B^{xz}(\ell,i)$,
        \item $B^x(\ell,i_\ell)[1,p_\ell]= B^{xz}(\ell,i_\ell)[1,p_\ell]$,
        \item $r_\ell - p_\ell + \sum_{j=i_\ell}^{t_\ell-1} |B^{xz}(\ell,j)| \le \ell (3R+3)+1$.
    \end{enumerate}
    
    For $\ell=0$, $B^x(0,1),B^x(0,2),\dots, B^x(0,s^x_0)=\Split(x,0)$ and  $B^{xz}(0,1),B^{xz}(0,2),\dots, B^{xz}(0,s^{xz}_0)=\Split(xz,0)$, 
    so we set $i_0=s^x_0$. If $|B^x(0,s^x_0)| \le |B^{xz}(0,s^x_0)|$ we set $t_0=i_0$, and $r_0=p_0=|B^x(0,i_0)|$, otherwise the last symbol of $x$ starts the block $B^{xz}(0,s^x_0+1)$ so we set $t_0=i_0+1$, $p_0= |B^{x}(0,s^x_0)|$ and $r_0=1$. (For completeness we set $i_{-1}=t_{-1}=1$.) Clearly, the four properties are satisfied by this choice.
    
    For $\ell=1,\dots,L$, having defined $i_{\ell-1},t_{\ell-1},p_{\ell-1},r_{\ell-1}$ we will define $i_{\ell},t_{\ell},p_{\ell}$, and $r_{\ell}$.
    As before, let $$A^x_{\ell-1}=\Compress(B^x(\ell-1,i_{\ell-1}),\ell) \;\;\;\;\&\;\;\;\; A^{xz}_{\ell-1}=\Compress(B^{xz}(\ell-1,i_{\ell-1}),\ell).$$

    \medskip\noindent {\em Case 1.}
    Consider the case when $i_{\ell-1}=t_{\ell-1}$. Let 
    \begin{align*}
        w &= B^{xz}(\ell-1,i_{\ell-1})[1,p_{\ell-1}],\\
        u_x &= B^{xz}(\ell-1,i_{\ell-1})[1+p_{\ell-1},r_\ell],\\
        u_z &= B^{xz}(\ell-1,i_{\ell-1})[1+r_{\ell-1},\dots],\\
        v &= B^{x}(\ell-1,i_{\ell-1})[1+p_{\ell-1},\dots].
    \end{align*}

    Let $A^{xz}_{\ell-1}=w'u'_x u'_z$ where $w'$ comes from the compression of $w$, 
    $u'_x$ comes from the compression of $u_x$, and $u'_z$ comes from the compression of $u_z$. 
    Let $A^{x}_{\ell-1}=w''v'$ where $w''$ comes from the compression of $w$, and $v'$ comes from the compression of $v$. 
    Set $r'_\ell=|w'u'_x|$, let $p'_\ell \le r'_\ell$ be the largest integer so that $A^{xz}_{\ell-1}[1,p'_\ell] = A^{x}_{\ell-1}[1,p'_\ell]$.
    By Lemma~\ref{l-extension}, $p'_\ell \ge |w'|-3(R+1)$ so $r'_\ell - p'_\ell \le |w'u'_x| - |w'| + 3(R+1) \le |u'_x| + 3(R+1)$.
    Furthermore, $|u'_x| \le |u_x| \le r_{\ell-1} - p_{\ell-1} \le (\ell-1) \cdot (3R+3)+1$, which follows by properties of $p_{\ell-1}$ and $r_{\ell-1}$, so $r'_\ell - p'_\ell \le \ell (3R+3)+1$.
    
    Let $(B_0,B_1,\dots,B_{m})=\Split(A^{xz}_{\ell-1},\ell).$
    Let $i\ge 0$ be the smallest integer such that $p'_\ell \le \sum_{j=0}^i|B_j|$ and 
    let $t\ge 0$ be the smallest such that $r'_\ell \le \sum_{j=0}^t|B_j|$. 
    Set $p_\ell = p'_\ell - \sum_{j=0}^{i-1}|B_j|$ and $r_\ell = r'_\ell - \sum_{j=0}^{t-1}|B_j|$.
    Let $i_\ell$ be the index of block $B_i$ among blocks $B^{xz}(\ell,1), B^{xz}(\ell,2), \dots, B^{xz}(\ell,s^{xz}_\ell)$, and 
    let $t_\ell$ be the index of $B_t$ among those blocks. 
    Notice, $B^x(\ell,j)= B^{xz}(\ell,j)$, for all $j<i_\ell$.
    We conclude the case by observing that 
    $r_{\ell} - p_{\ell} + \sum_{j=i}^{t-1}|B_j| = \sum_{j=i}^{t-1}|B_j| + (r'_\ell - \sum_{j=0}^{t-1}|B_j|) - (p'_\ell - \sum_{j=0}^{i-1}|B_j|)
    = p'_\ell - r'_\ell \le \ell (3R+3)+1$.
     
    \medskip\noindent {\em Case 2.}
    The case $i_{\ell-1}<t_{\ell-1}$ is similar. 
    In this case we let $w$ and $v$ to be as in the previous case and $u=B^{xz}(\ell-1,i_{\ell-1})[1+p_{\ell-1},\dots]$.
    We let $p'_\ell \le |A^{xy}_{\ell-1}|$ be the largest integer so that $A^{xz}_{\ell-1}[1,p'_\ell] = A^{x}_{\ell-1}[1,p'_\ell]$.
    By Lemma~\ref{l-extension}, $p'_\ell \ge |A^{xy}_{\ell-1}|-3(R+1)-|u| = |A^{xy}_{\ell-1}|-3(R+1)- |B^{xz}(\ell-1,i_{\ell-1})|+p_{\ell-1}$.  
    Rearranging terms: $-p'_\ell+|A^{xy}_{\ell-1}| \le - p_{\ell-1} + |B^{xz}(\ell-1,i_{\ell-1})| + 3(R+1)$.
    
     Let $i\ge 0$ be the smallest integer such that $p'_\ell \le \sum_{j=0}^i|B_j|$.
     Let $p_\ell = p'_\ell - \sum_{j=0}^{i-1}|B_j|$, 
     and $i_\ell$ be the index of the block $B_i$ within $B^{xz}(\ell,1), B^{xz}(\ell,2), \dots, B^{xz}(\ell,s^{xz}_\ell)$.
    Hence, $-p_\ell +  \sum_{j=i_\ell}^{i_\ell+m-i} |B^{xz}(\ell,j)| = -p_\ell +  \sum_{j=i}^{m} |B_j| = -p'_\ell+  \sum_{j=0}^{m} |B_j| =  -p'_\ell+|A^{xy}_{\ell-1}| \le - p_{\ell-1} + |B^{xz}(\ell-1,i_{\ell-1})| + 3(R+1)$.
    
    Let $C^{xz}_{\ell-1}=\Compress(B^{xz}(\ell-1,t_{\ell-1}),\ell)$ and $(B'_0,B'_1,\dots,B'_{m'})=\Split(C^{xz}_{\ell-1},\ell)$.
    Let $r'_\ell$ be the largest position in $C^{xz}_{\ell-1}$ of a symbol coming from compression of $x$, 
    and $t\ge 0$ be the smallest integer such that $r'_\ell \le \sum_{j=0}^t|B_j|$, and set $r_\ell = r'_\ell - \sum_{j=0}^{t-1}|B'_j|$. 
    Let $t_\ell$ be the index of the block $B'_t$ within $B^{xz}(\ell,1), B^{xz}(\ell,2), \dots, B^{xz}(\ell,s^{xz}_\ell)$. 
    Clearly, $r'_\ell \le r_{\ell-1}$ so 
    $r_\ell + \sum_{j=t_\ell-t}^{t_\ell-1} |B^{xz}(\ell,j)| = r_\ell + \sum_{j=0}^{t-1}|B'_j| \le r'_\ell \le r_{\ell-1}$.
    
    Notice,  $\sum_{j=i_\ell+m-i+1}^{t_\ell-t-1} |B^{xz}(\ell,j)| \le  \sum_{j=i_{\ell-1}+1}^{t_{\ell-1}-1} |B^{xz}(\ell-1,j)|$.
    By partitioning the sum, rearranging the terms and using the upper bounds derived so far we have: 
    $r_\ell - p_\ell + \sum_{j=i_\ell}^{t_\ell-1} |B^{xz}(\ell,j)| 
    = -p_\ell +  \sum_{j=i_\ell}^{i_\ell+m-i} |B^{xz}(\ell,j)| + \sum_{j=i_\ell+m-i+1}^{t_\ell-t-1} |B^{xz}(\ell,j)| + \sum_{j=t_\ell-t}^{t_\ell-1} |B^{xz}(\ell,j)| + r_\ell 
    \le - p_{\ell-1} + |B^{xz}(\ell-1,i_{\ell-1})| + 3(R+1) +  \sum_{j=i_{\ell-1}+1}^{t_{\ell-1}-1} |B^{xz}(\ell-1,j)| + r_{\ell-1} 
    = r_{\ell-1} - p_{\ell-1} + \sum_{j=i_{\ell-1}}^{t_{\ell-1}-1} |B^{xz}(\ell-1,j)| + 3(R+1)
    \le (\ell-1)\cdot (3R+3)+1 +  3(R+1) \le \ell (3R+3) + 1 $, where the second to last inequality follows by the properties of our numbers for $\ell-1$.
    
    For $\ell=L$ we get: $r_L - p_L + \sum_{j=i_L}^{t_L-1} |B^{xz}(L,j)| \le L (3R+3)+1$. Since $p_L \le |B^{xz}(L,i_L)|$, and $r_L \ge 0$ we get: $\sum_{j=i_L+1}^{t_L-1} |B^{xz}(\ell,j)| \le L (3R+3)+2$. Since each $B^{xz}(L,j)$ is of non-zero size,
    $t_L-i_L-1 \le L (3R+3)+2$. Thus $t_L \le i_L +  L (3R+3)+3 \le s^x_L + L (3R+3)+3$, as $i_L \le s^x_L$. The claim follows.
\end{proofof}

\section{Appending a symbol to a grammar decomposition}
\label{s-update}

In this section we provide a detailed description of the process of updating the active grammars of a string $x$ when appending a new symbol $a$. 
By Lemma~\ref{l-grammarsuffix} only the last $T$ grammars of $x$ might change when adding a new symbol $a$.
As observed already previously, Lemma~\ref{l-grammarsuffix} also implies that once a grammar becomes more than
$(T+1)$-th grammar from the end it will never change, despite the fact that the number of grammars that follow it might shrink after adding more symbols. 
(Adding more symbols might create periodicity that will be exploited by the compression.)
Our rolling sketch algorithm keeps at most $T$ active grammars that might still change after adding more symbols.
It is convenient for our implementation of the update function to have access also to the previous at most $T$ committed grammars (to have the proper context for re-compression).
Our rolling sketch algorithm has those committed grammars available in appropriate buffers.
Thus we will assume that the update function is always invoked with exactly $T+1$ grammars, unless $x$ is decomposed into less that $T+1$ grammars. 
Some of the first few grammars from the output of the update procedure should be discarded as they correspond to grammars that should stay the same.
In particular, if there are $t$ active grammars and $s$ committed grammars then we should discard the first $\min(s,T+1-t)$ grammars
from its output.
The following statement encapsulates the properties of our update procedure $\UpdateAG()$.

\begin{theorem}\label{t-update}
    Let integers $k\le n$ and functions $C_1,\dots,C_L$ and $H_0,\dots,H_L$ be given.
    For any $a\in \Sigma$ and $x\in \Sigma^*$ of length at most $n$ with $G_1,\dots,G_s$ being the grammars output by the decomposition algorithm on input $x$ using functions $C_1,\dots,C_L,H_0,\dots,H_L$, $\UpdateAG(G_{s-\min(s,T+1)+1},\dots,G_s,a)$ outputs a sequence of grammars $G'_1,\dots,G'_{t'}$
    such that $G_1,\dots,G_{s-\min(s,T+1)}, G'_1,\dots,G'_{t'}$ is the sequence that would be output by the decomposition algorithm on $x\cdot a$  using the functions $C_1,\dots,C_L, H_0,\dots,H_L$.
    The update algorithm runs in time $\OO(k L T)=\OO(k)$ and outputs $t' \le 4TL$ grammars.
\end{theorem}

Here we assume that the decomposition algorithm does not fail neither on $x$ nor on $x\cdot a$
with respect to producing correct deterministic grammars 
so the first two parts of Theorem~\ref{t-decomposition} are satisfied for $x$ and $y=x\cdot a$, 
and the choice of functions $C_1,\dots,C_L$ and $H_0,\dots,H_L$.
For the simplicity of our implementation, we assume a stronger property of $C_1,\dots,C_L$, 
that each $C_\ell$ is one-to-one on the union of all blocks of $x$ and $x \cdot a$ at level $\ell$.
(See remark after Lemma~\ref{l-onetoone}.)

\subsection{Auxiliary functions}

Our update algorithm uses several simple and straightforward auxiliary functions we describe next.
Function $\DecompressSymbol(c,G,\ell,t)$ takes a symbol $c\in \Gamma$ and if it is a level-$\ell$ symbol
compressed by the grammar $G$ then it returns its decompression truncated to the length of at most $t$ symbols. 
Otherwise it returns the original symbols $c$.

\begin{algorithm}[H]
\DontPrintSemicolon
   \caption{$\DecompressSymbol(c,G,\ell,t)$}
   \KwIn{A symbol $c$, a grammar $G$, a level $\ell$, maximum output size $t\ge 2$.}
   \KwOut{Decompresses $c$ if it was compressed at level $\ell$. Returns at most $t$ symbols of the decompression.}
   
   \vspace{1mm}
   \hrule\vspace{1mm}

    \lIf{$c\in \Sigma^\ell_c$}{ let $a,b \in \Gamma$ be such that $c\rightarrow ab \in G$. Return $ab$. }
    \lIf{$c\in \Sigma^\ell_r$}{ let $a\in \Gamma,r\in \N$ be such that $c=\rr_{a,r}$. Return $a^{\min(t,r)}$. }

    Return $c$.

\end{algorithm}

Function $\DecompressString(Z,G,\ell)$ decompresses all level-$\ell$ compression symbols in a string $Z\in \Gamma^*$ using the grammar $G$, and returns the resulting decompressed string. 

\begin{algorithm}[H]
\DontPrintSemicolon
   \caption{$\DecompressString(Z,G,\ell)$}
   \KwIn{A string $Z$, a grammar $G$, and level $\ell$.}
   \KwOut{Decompresses $z$ at level $\ell$.}
   
   \vspace{1mm}
   \hrule\vspace{1mm}

    $Y=\eps$.

    \lFor{$i=1$ to $|Z|$}{ $Y = Y\cdot \DecompressSymbol(Z[i],G,\ell,\infty)$.}

    Return $Y$.

\end{algorithm}

Function $\DecompressSymbolLength(c,\ell)$ returns the length of the decompression of a symbol $c$ at level $\ell$.

\begin{algorithm}[H]
\DontPrintSemicolon
   \caption{$\DecompressSymbolLength(c,\ell)$}
   \KwIn{A symbol $c$, a level $\ell$.}
   \KwOut{Returns the length of decompression of $c$ at level $\ell$.}
   
   \vspace{1mm}
   \hrule\vspace{1mm}

    \lIf{$c\in \Sigma^\ell_c$}{ return 2. }
    \lIf{$c\in \Sigma^\ell_r$}{ let $a \in \Gamma,r\in \N$ be such that $c=\rr_{a,r}$. Return $r$. }

    Return $1$.

\end{algorithm}

% Function $\FindEqualCompressedPrefix(Z,Z',\ell)$ finds the length of the prefix of $Z$ that at level $\ell$ 
% gives the same decompression as $Z'$.

% \begin{algorithm}[H]
% \label{algo:find_index_prefix}
% %\caption{Compression}
% \DontPrintSemicolon
%    \caption{$\FindEqualCompressedPrefix(Z,Z',\ell)$}
%    \KwIn{Strings $Z$ and $Z'$ where the decompression of $Z'$ is a prefix of the decompression of $Z$ at level $\ell$.}
%    \KwOut{An index $j$ such that level $\ell$ decompression of $Z[1,j]$ and $Z'$ are the same.}
   
%    \vspace{1mm}
%    \hrule\vspace{1mm}
    
%     $p = 0$. \acc{Compute the length of decompression of $Z'$.}
    
%     \For{$i = 1$ to $|Z'|$}{ 
%         $p = p+\DecompressSymbolLength(Z'[i],\ell)$;
%     }
    
%     Return $\FindCompressedPrefix(Z,p,\ell)$. 

% \end{algorithm}

\begin{algorithm}[H]
\label{algo:compress1}
%\caption{Compression}
\DontPrintSemicolon
   \caption{$\CompressWithGrammar(B,\ell)$}
   \KwIn{String $B$ over alphabet $\Gamma$, and level number $\ell$.}
   \KwOut{String $B''$ over alphabet $\Gamma$, and set of applied rules $G'$.}
   
   \vspace{1mm}
   \hrule\vspace{1mm}

   \lIf{$|B|\le 1$}{ return $B,\emptyset$. }

    Set $G'=\emptyset$.
   
    Divide $B=B_1B_2B_3\dots B_m$ into minimum number of blocks so that each maximal subword $a^r$ of $B$, for $a\in\Gamma$ and $r\ge 2$, is one of the blocks.
    
    \For{each $i\in \{1,\dots, m\}$}{
      \If{$B_i=a^r$, where $r\ge 2$}{
         Set $B'_i = \rr_{a,r} \cdot \#$ and color $\rr_{a,r}$ by 1 and $\#$ by 2.
         
        $G' = G' \cup \{ \rr_{a,r} \rightarrow a^r \}$;
      }
      \lElse{Set $B'_i = B_i$ and color each symbol of $B'_i$ according to $\FL(B_i)$}
    }
    
    Set $B'=B'_1B'_2\cdots B'_m$, $B''=\eps$, and $i = 1$.

    \While{$i<|B'|$}{

        \lIf{ $B'[i+1] = \#$}{ $B'' = B'' \cdot B'[i]$}
        \Else{ 
            $B'' = B'' \cdot C_\ell(B'[i,i+1])$;

            $G' = G' \cup \{ C_\ell(B'[i,i+1])\rightarrow B'[i,i+1] \}$;
         }

        $i = i+2$.

        \lIf{$i\le |B'|$ and $B'[i]$ is not colored 1}{
            $B'' = B'' \cdot B'[i]$, $i = i+1$    
        }
    }

    Return $B'', G'$. 

\end{algorithm}

Function $\CompressWithGrammar(B,\ell)$ is an extension of $\Compress(B,\ell)$ that in addition
to compressed block $B$ at level $\ell$ returns the set of grammar rules used for the compression of $B$ at this level.

Finally, function $\FindCompressedPrefix(Z,p,\ell)$ returns the length of the smallest prefix of a string $Z$
that decompresses into at least $p$ symbols at level $\ell$.

\begin{algorithm}[H]
\label{algo:find_index}
%\caption{Compression}
\DontPrintSemicolon
   \caption{$\FindCompressedPrefix(Z,p,\ell)$}
   \KwIn{String $Z$, an integer $p$, level $\ell$.}
   \KwOut{Smallest index $j$ such that level $\ell$ decompression of $Z[1,j]$ has length $\ge p$.}
   
   \vspace{1mm}
   \hrule\vspace{1mm}
    
    $q=0$ and $j=0$.
    
    \While{$q < p$}{
        $j = j+1$;

        $p = p+\DecompressSymbolLength(Z[j],\ell)$;
    }
    
    Return $j$. 

\end{algorithm}

\subsection{Main functions}

The core of the update function $\UpdateAG((G_1,\dots,G_t),a)$ is build around the functions we describe next.
The functions use globally accessible set of grammar rules $G$ that contains all the rules from
$G_1,\dots,G_t$ except for the starting rules. 
(This set of rules is deterministic assuming the remark after Theorem~\ref{t-update}.)

The functions will build a sequence of strings $Z_L,Z_{L-1},\dots,Z_0$ each of length at most $2T$.
$Z_L$ is the concatenation of the right-hand-sides of starting rules of $G_1,\dots,G_t$.
For $\ell=L,\dots,1$, $Z_{\ell-1}$ is then build inductively by decompressing a (suitable) largest suffix 
of $Z_\ell$ so that $Z_{\ell-1}$ would be of length at most $T+4 \le 2T$.
The decompression is provided by function $\PartialDecompress(Z,F,\ell)$ which returns tuple $Z',F',u,r'$.
In the case that the first symbol of the decompressed suffix of $Z_\ell$ is the level-$\ell$ repeat symbol
$\rr_{a,r}$ that would expand $Z_{\ell-1}$ beyond the limit of $T+4$ symbols, we truncate the expansion
of that symbol to the length $r_{\ell}=r'$. 
The return value $u$ indicates how many symbols of $Z$ were left uncompressed (which would include the partially decompressed symbol $\rr_{a,r}$). 
It satisfies that if $u\ne 0$ then $|Z'|\ge T$.
Strings $Z_L,\dots,Z_0$ satisfy that for $\ell=L,\dots,1$, if $|Z_\ell|\ge T$ then $|Z_{\ell-1}|\ge T$.
(In particular, if $\UpdateAG()$ is invoked with at least $T+1$ grammars, then all $Z_\ell$ are of length at least $T$.
The compression of the first grammar might depend on unseen grammars in that case so we cannot re-compress it at will.)

Strings $Z_L,\dots,Z_0$ are accompanied by strings of integers $F_L,\dots,F_0$ over the alphabet $\{0,\dots,L+1\}$. 
The value of $F_\ell[i]$ indicates at which level the symbol $Z_\ell[i]$ becomes the first symbol in its block.
In particular, $F_\ell[i] < \ell$ indicates that a block starts at position $i$ of $Z_\ell$.
This value is relevant for re-compression of updated strings $Z_0,\dots,Z_L$.
The initial values of $F_L$ are computed using $\SplittingDepth(G)$.
Function $\SplittingDepth(G)$ is fairly straightforward: For a grammar $G$, it inductively decompresses
the first two symbols of the evaluation of $G$. 
It finds the lowest level $\ell$, at which the first two symbols of the decompression give zero when function $H_\ell$ is applied on them.  

After obtaining $Z_L,\dots,Z_0$, $\UpdateAG((G_1,\dots,G_t),a)$ appends $a$ to $Z_0$, and then
re-compresses $Z_0,\dots,Z_{L-1}$ using a function $\Recompress(B,Z,F,u,r,\ell)$. 
We provide more details on function $\Recompress(B,Z,F,u,r,\ell)$ further below.
Invoking $\UpdateAG((G_1,\dots,G_t),a)$ returns a sequence of updated grammars.

\begin{algorithm}[H]
\label{algo:partial_decompress}
%\caption{Compression}
\DontPrintSemicolon
   \caption{$\PartialDecompress(Z,F,\ell)$}
   \KwIn{String $Z$, splitting depth string $F$, and level $\ell$.}
   \KwOut{Decompressed string $Z'$, splitting depth string $F'$, unused count $u$, repeat count $r'$.}
   
   \vspace{1mm}
   \hrule\vspace{1mm}

    Set $Z'=\eps$ and $F'=\eps$.
    
    \For{$u = |Z|$ to $1$}{
        
        \If{$Z[u] = \rr_{a,r}$, where $\rr_{a,r} \in \Sigma_{r}^{\ell}$}{\;
           \lIf{$|Z'| + r \le T + 3$}{
                $Z' = a^r \cdot Z'$ and
                $F' = F[u] \cdot (L+1)^{r-1} \cdot F'$
            }
            \Else{
                $r' = T - |Z'| + 1$; 
                
                $Z' = a^{r'} \cdot Z'$ and $F' = (L+1)^{r'} \cdot F'$ ;
 
                Return $Z',F',u,r'$.
            }
        }
        \ElseIf{$Z[u]=a$, where $a \in \Sigma_{c}^{\ell}$}{
            $Z' = b\cdot c \cdot Z'$, where $a\rightarrow b\cdot c$ is in $G$;

            $F' = F[u] \cdot (L+1) \cdot F'$;
        }
        \lElse{$Z' = Z[u] \cdot Z'$ and $F' = F[u] \cdot F'$}
         
        \lIf{$|Z'| \ge T$}{
             return $Z',F',u-1,0$.                 
        }
    }

    Return $Z',F',0,0$.     
\end{algorithm}

\begin{algorithm}[H]
\DontPrintSemicolon
   \caption{$\SplittingDepth(G)$}\label{alg-split-depth}
   \KwIn{Non-empty grammar $G$.}
   \KwOut{The first level $\ell$ where $G$ would be separated as a new block.}
   
   \vspace{1mm}
   \hrule\vspace{1mm}

    Let $v$ be such that $\# \rightarrow v \in G$. \acc{$v$ are the first two symbols of $\eval(G)$.}

    $d = L+1$. 
    
    \For{ $\ell=L,\dots, 0$}
    {
        \lIf{$|v|\ge 2$ and $H_\ell(v[1,2])=0$}{$d = \ell$.}

        $u = \DecompressSymbol(v[1],G,\ell,2)$

        \lIf{$|v| \ge 2$}{ $u = u \cdot \DecompressSymbol(v[2],G,\ell,2)$.}
        
        $v=u$.
    }

    Return $d$.

\end{algorithm}

\begin{algorithm}[H]
\label{algo:update_ag}
%\caption{Compression}
\DontPrintSemicolon
   \caption{$\UpdateAG(AG,a)$}
   \KwIn{List of grammars $AG = (G_1 , \dots , G_t)$ representing a string $x$, and a symbol $a$.}
   \KwOut{Updated list of grammars $AG'$ representing string $x\cdot a$.}
   
   \vspace{1mm}
   \hrule\vspace{1mm}

    \acc{Construct a set of rules $G$, initial compressed string $Z_L$ and splitting depth string $F_L$.} 

    For $i=1,\dots,t$, let $\# \rightarrow v_i$ be the starting rule in $G_i$.

    Set $G = \bigcup_{i=1}^t G_i \setminus \{\# \rightarrow v_i\}$. 

    Set $Z_L = v_1$ and $F_L = 0 \cdot (L+1)^{|v_1|-1}$.

    For $i=2,\dots,t$, set $Z_L = Z_L \cdot v_i$ and $F_L = F_L \cdot \SplittingDepth(G_i) \cdot (L+1)^{|v_i|-1}$.
    
    \acc{Perform partial decompression}
        
    \For{$\ell = L$ to $1$}{
        $Z_{\ell-1}, F_{\ell-1}, u_{\ell}, r_{\ell} = \PartialDecompress(Z_{\ell},F_{\ell}, \ell)$.
    }

    \acc{Perform re-compression}

    $Z_0 = Z_0\cdot a$;
    $B = \Split(Z_0,0)$;
    
    \For{$\ell = 1$ to $L$}{
        $B',G' = \Recompress(B,Z_{\ell}, F_{\ell}, u_{\ell}, r_{\ell}, |Z_{\ell-1}|, \ell)$

        $G = G \cup G'$
        
        $B = B'$        
    }

    Let $B = (B_1,\dots , B_{t'})$.

    $AG'=()$.
    
    \For{$i = 1$ to $t'$}{
        $G' = G \cup \{\#\rightarrow B_i\}$.
        
        Remove from $G'$ unnecessary rules to get $G'_i$ (as in Section~\ref{s-grammars}).
        
        Append $G'_i$ to $AG'$.
    }
    
    Return $AG'$. 

\end{algorithm}

Function $\Recompress(B,Z,F,u,r,\ell)$ gets a sequence $B=(B_0,\dots,B_s)$ of blocks
that represent compression of the updated $Z_{\ell-1}$ (after adding $a$) up-to level $\ell-1$. 
It also gets the original $Z_\ell$, the splitting depth string $F_\ell$, the number of symbols $u_\ell$
that were decompressed from $Z_\ell$ to get the original $Z_{\ell-1}$ and the parameter $r_\ell$
that indicates that the first $r_\ell$ symbols of $Z_{\ell-1}$ are a partial decompression of the repeat symbol $Z_{\ell}[u]$. 
It outputs a sequence of blocks $B'$ that represent the updated block $Z_\ell$ compressed up-to level $\ell$, and a set of rules $G'$ that were used for compression at level $\ell$.

Blocks $B_1,\dots,B_s$ can be independently compressed and split at level $\ell$.
The block $B_0$ needs a special treatment though as it needs to be combined with its possible remainder in $Z_\ell$.  
This is done in function $\RecompressFirstBlock(B_0,Z,F,u,r,\ell)$.
Remaining blocks for the output $\Recompress()$ are obtained from $Z_\ell$ by splitting it into blocks according 
to $F_\ell$.

\begin{algorithm}[H]
\label{algo:recompression}
%\caption{Compression}
\DontPrintSemicolon
   \caption{$\Recompress(B,Z,F,u,r,z,\ell)$}
   \KwIn{$B=(B_0,\dots,B_s)$ sequence of blocks, original uncompressed string $Z$, splitting depth string $F$ of $Z$, $u$ number of uncompressed symbols in $Z$, repeat count $r$, $z=|Z_{\ell-1}|$, and level $\ell$.}
   \KwOut{$B'$ a new sequence of blocks representing $B$ together with $Z[1,u]$, and set of newly added rules $G'$.}
   
   \vspace{1mm}
   \hrule\vspace{1mm}

    \lIf{$z<T$}{ $B'=()$, $G'=\emptyset$, $u'=0$, $j=0$. \acc{No symbols precede $B_0$.} }
    \Else{
       $B',G',u' = \RecompressFirstBlock(B_0,Z,F,u,r,\ell)$.     \acc{ Compress block $B_0$. }
       
       $j=1$.
     }
    
    \acc{ Compress blocks $B_j,\dots,B_s$.}
   
    \For{$i = j$ to $s$}{  
        \lIf{$|B_i|\le 2$}{ $B'' = (B_i)$; $G''=\emptyset$ }
        \Else{
            $B'_i,G'' = \CompressWithGrammar(B_i,\ell)$.

            $B'' = \Split(B'_i,\ell)$.
        }

        Append $B''$ to $B'$.

        $G'= G' \cup G''$.
    }

    $i=u'$.    \acc{ Separate remaining blocks in $Z$. }

    \While{$i>0$}{
        \lWhile{$i>1$ and $F[i]>\ell$}{$i=i-1$.}

        Add $Z[i,u']$ as the first item of $B'$.

        $i=i-1; u'=i$.
        
    }

    Return $B',G'$. 
\end{algorithm}

Function $\RecompressFirstBlock(B_0,Z,F,u,r,\ell)$ is the most complicated function of the whole re-compression process. 
The function is invoked only if $|Z_{\ell-1}|\ge T$.
The function gets the first level $\ell-1$ block $B_0$ that needs to be combined with its remainder in $Z=Z_\ell$. 
The remainder is a suffix of $Z_\ell[1,u]$, where $r$ indicates that the first $r$ symbols of the original $Z_{\ell-1}$ were obtained by the partial decompression of $Z_\ell[u]$. 
If $r\ne 0$ then the compression of the part of $B_0$ that follows its leading $a$'s ($Z_\ell[u]=\rr_{a,r'}$) is independent of the compression of the part of $Z$ belonging to $B_0$ and preceding $Z_\ell[u]$, as $r'-r\ge 2$. 
Thus we can compress that part of $B_0$, combine it with an appropriate repetition symbol $\rr_{a,r''}$
and append it to the appropriate suffix of $Z_\ell[1,u-1]$ (which is already compressed at level $\ell$.)
% If $u=0$ then there is nothing to append to $B_0$ so we just compress it.
If $r=0$ then we invoke a function $\CrossOverBlock(B_0,Z[u',\dots],u-u'+1,\ell)$, where $u'$ is the first symbol in $Z_\ell$ that belongs to the block of $B_0$.
Eventually, we split the compressed block $B_0$ using $\Split()$.

\begin{algorithm}[H]
\label{algo:recompressionFirst}
%\caption{Compression}
\DontPrintSemicolon
   \caption{$\RecompressFirstBlock(B_0,Z,F,u,r,\ell)$}
   \KwIn{Block $B_0$, an original uncompressed string $Z$, splitting depth string $F$ of $Z$, $u$ number of uncompressed symbols in $Z$, repeat count $r$, and level $\ell$.}
   \KwOut{$B'$ a new sequence of blocks representing $B_0$ together with $Z[1,u]$, and set of newly added rules $G'$, number $u'$ of unused symbols in $Z$.}

    \lIf{$r\ne 0$}{$u=u-1$.}       

    $u'= u+1$. \acc{ Find the beginning of block $B_0$ in the uncompressed part $Z$. }

    \lWhile{$u'>1$ and $d[u'] \ge \ell$}{
        $u'=u'-1$.
    }
    \If{$r \ne 0$}{
        \acc{ Block $B_0$ starts by partially decompressed symbol $\rr_{a,r}$. }
    
        Let $a\in \Gamma$ and $r'\in \N$ be such that $Z[u+1]=\rr_{a,r'}$.

        \lFor{$i=1$ to $|B_0|$}{
            \lIf{$B_0[i] \ne a$}{break;}
        }
        \lIf{$B_0[i]=a$}{ $B'=\eps, G'' = \emptyset$, $i=i+1$. }
        \lElse{$B',G'' = \CompressWithGrammar(B_0[i,\dots],\ell)$.}

        $B' = Z[u',u] \cdot \rr_{a,r'-r+i-1} \cdot B'$.    
    }
%    \ElseIf{$u=0$}{$B',G'' = \CompressWithGrammar(B_0,\ell)$.}
    \Else{
        $B',G'' = \CrossOverBlock(B_0, Z[u',\dots],u-u'+1,\ell)$.
    }

    $B'' = \Split(B',\ell)$.

    Return $B'',G'',u'-1$. 
\end{algorithm}

Function $\CrossOverBlock(B,Z,u,\ell)$ gets a block $B$ that was compressed up-to level $\ell-1$ and 
needs to be combined with its remainder $Z[1,u]$ that is compressed up-to level $\ell$. 
(The resulting block should correspond to ``$Z[1,u] \cdot B$''.) 
We know that $|Z_{\ell-1}|\ge T \ge L(3R+3)$ otherwise $\RecompressFirstBlock()$ and $\CrossOverBlock()$ would not be called.
By the three properties of $\Delta_{\ell-1}$ and $i_{\ell-1}$ defined in the proof of Part 1 of Lemma~\ref{l-grammarsuffix}
we know that the first $3(R+1)$ symbols of $Z_{\ell-1}$ were not modified as a result of appending the new symbol to $x$. 
Hence the first $\min(3(R+1),|B|)$ symbols of $B$ correspond to the decompression of $Z[u+1,\dots]$. 

In this part of $B$ we look for any repeated symbol. 
If we find a repeated symbol there, we combine the compression of the part of $B$ starting at the repeated symbol with the original
part of $Z[u+1,\dots]$ that produced the symbols of $B$ preceding the repeated symbol (and also with $Z[1,u]$).
By the properties of compression, repeated symbols break dependence between compressed symbols.

If $|B|\le 2R+20$ then at least $3(R+1)-2R-20>2$ unchanged symbols follow $B$. 
Thus $B$ ends at its original location as it was split at some level $<\ell$ and the first two symbols of the next block at all
levels $<\ell$ are the same as originally.

Finally, if $|B|>2R+20$ and there is no repeated symbol in the first up-to $3(R+1)$ symbols of $B$
then we can compress $B$ to get $B'$, strip from $B'$ the compression of the first $R+10$ symbols 
and combine it with the original compression of those $R+10$ symbols from $Z$.
(The first up-to $3(R+1)$ symbols of $B$ consist of singletons. 
The compression of a singleton depends on the context of at most $R+3$ symbols on either side.)  

\begin{algorithm}[H]
\label{algo:recompressionCross}
%\caption{Compression}
\DontPrintSemicolon
   \caption{$\CrossOverBlock(B,Z,u,\ell)$}
   \KwIn{Block $B$, an original uncompressed string $Z$, number $u$ of unused symbols in $Z$, and level $\ell$.}
   \KwOut{$B'$ and set of newly added rules $G'$.}

  \acc{ Try to find a repeated symbol in unmodified $B$. }

   $i=1$.

   \lWhile{$i<|B|$ and $i < 3(R+1)$ and $B[i]\ne B[i+1]$}{ $i=i+1$.}

   \If{$i<|B|$ and $B[i] = B[i+1]$}{
       \acc{$B[i]$ is a repeated symbol.}
   
        $B',G' = \CompressWithGrammar(B[i,\dots],\ell)$.

        $j = \FindCompressedPrefix(Z[u+1,\dots],i-1,\ell)$.

        $B' = Z[1,u+j] \cdot B'$.
    }
    \ElseIf{$|B|\le 2R+20$}{
        $j = \FindCompressedPrefix(Z[u+1,\dots], |B|, \ell)$. \acc{At least two unchanged symbols follow $B$.}

        $B' = Z[1,u+j]$, $G'=\emptyset$.
    }
    \Else{
        $B',G' = \CompressWithGrammar(B,\ell)$.

        $p = \FindCompressedPrefix(B',R+10, \ell)$.

        $j = \FindCompressedPrefix(Z[u+1,\dots], R+10, \ell)$.

        $B' = Z[1,u+j-1] \cdot B'[p,\dots]$.
     }

     Return $B',G'$. 
\end{algorithm}

The correctness of the update algorithm follows from its description.

\subsection{Time analysis}

We assume that strings are represented efficiently (e.g. by balanced trees) so we can extract a sub-string, concatenate strings, etc. in time $\OO(1)$. All strings that we will operate on will be of length $O(T)$.
Similarly, we assume that grammars are represented efficiently so that we can look-up a rule with a given 
left-hand symbol, append two grammars, etc. in time $\OO(1)$. 

Then $\DecompressSymbol()$ and $\DecompressSymbolLength()$ takes time $\OO(1)$.
The time complexity of each of the functions $\CompressWithGrammar()$, $\DecompressString()$, 
% $\FindEqualCompressedPrefix()$, 
$\FindCompressedPrefix()$, $\PartialDecompress()$ and $\CrossOverBlock()$ is proportional to the length of strings on which it operates so it is $\OO(T)$.
Time of $\SplittingDepth()$ is proportional to the depth of the grammar, which in our case is at most $\OO(L)$.
Each $\RecompressFirstBlock()$ executes $O(T)$ operations on strings and grammars, and $O(T)$
evaluations of $H_\ell$ (inside the calls to $\Split()$). 
% Since $H_\ell$ is $\OO(k)$-wise independent, its evaluation takes time $\OO(k)$.
So $\RecompressFirstBlock()$ takes time $\OO(T)$.

Similarly, each $\Recompress()$ executes up-to one call to $\RecompressFirstBlock()$, $O(T)$ operations on strings and grammars, and $O(T)$ evaluations of $H_\ell$ to split blocks. 
Again, its total time complexity is $\OO(T)$.
Eventually, $\UpdateAG()$ executes up-to $T$ $\SplittingDepth()$, $O(T)$ string operations, 
$L$ calls to  $\PartialDecompress()$ and $\Recompress()$, and then up-to $O(LT)$ invocations of 
grammar minimization procedure costing $\OO(k)$ time each. Thus, the total time for $\UpdateAG()$ is $\OO(LTk)$. 

The number of grammars the algorithm outputs is at most $\sum_{\ell=0}^L |Z_\ell| \le 2T(L+1) \le 4TL$.

\section{Table of parameters}\label{s-param}

\begin{center}
\begin{tabular}{ |l|c|c|c| } 
 \hline
 Definition & Asymptotics & Meaning & Reference \\
 \hline
 \hline
 $R = \log^{*}|\Gamma| + 20$ & $\log^* n$ & compression locality & Section \ref{s-locally_consistent_parsing}\\ \hline
 $L=\lceil \log_{3/2} n \rceil+3$ & $ \log n$ & recursion depth & Section \ref{s-decomposition}, Corollary \ref{c-decomposition1}\\ \hline
 $D = 110R(L+1)k$ & $k \log n \log^* n$ & $1/$splitting probability & Section \ref{s-decomposition}, Lemma \ref{l-decomposition1}\\ \hline
 $S=15 DL \log n + 3$ & $k \log^3 n \log^* n$ & maximum grammar size & Section \ref{s-decomposition}, Theorem \ref{t-decomposition}\\ \hline
 $M= 3 S \cdot \lceil 1 + \log |\Gamma| \rceil$ &  $k \log^4 n \log^* n$ & grammar encoding size & Section \ref{s-binencoding}\\ \hline
 $T=L(3R+6)$ & $\log n \log^* n$ & locality of suffix changes & Section \ref{s-rolling}, Lemma \ref{l-grammarsuffix}\\ \hline
 $N \ge n^3$ & $n^3$  & $\FKR$ range size & Section \ref{s-binencoding}\\ \hline
\end{tabular}
\end{center}

\section*{Acknowledgements}

The authors benefited greatly from discussions with Nicole Wein who took part in the initial stages of this project. 
The second author also benefited from many discussions on edit distance and on hash functions with Mike Saks. 
We are grateful to Tomasz Kociumaka for providing us with a reference for Proposition~\ref{p-edgrammar}.
We thank anonymous reviewers for their comments.

\bibliographystyle{alpha}
\bibliography{edit_sketch}

\newcommand{\etalchar}[1]{$^{#1}$}
\begin{thebibliography}{CDG{\etalchar{+}}18}

\bibitem[AN20]{AN20}
Alexandr Andoni and Negev~Shekel Nosatzki.
\newblock Edit distance in near-linear time: it's a constant factor.
\newblock In Sandy Irani, editor, {\em 61st {IEEE} Annual Symposium on
  Foundations of Computer Science, {FOCS} 2020, Durham, NC, USA, November
  16-19, 2020}, pages 990--1001. {IEEE}, 2020.

\bibitem[BES06]{BES06}
Tu\u{g}kan Batu, Funda Ergun, and Cenk Sahinalp.
\newblock Oblivious string embeddings and edit distance approximations.
\newblock In {\em Proceedings of the Seventeenth Annual ACM-SIAM Symposium on
  Discrete Algorithm}, SODA '06, pages 792--801, Philadelphia, PA, USA, 2006.
  Society for Industrial and Applied Mathematics.

\bibitem[BGP20]{lcp_application2020}
Or~Birenzwige, Shay Golan, and Ely Porat.
\newblock Locally consistent parsing for text indexing in small space.
\newblock In Shuchi Chawla, editor, {\em Proceedings of the 2020 {ACM-SIAM}
  Symposium on Discrete Algorithms, {SODA} 2020, Salt Lake City, UT, USA,
  January 5-8, 2020}, pages 607--626. {SIAM}, 2020.

\bibitem[BI15]{BI15}
Arturs Backurs and Piotr Indyk.
\newblock Edit distance cannot be computed in strongly subquadratic time
  (unless {SETH} is false).
\newblock In {\em Proceedings of the Forty-Seventh Annual ACM on Symposium on
  Theory of Computing}, STOC '15, pages 51--58, New York, NY, USA, 2015. ACM.

\bibitem[BR20]{BR20}
Joshua Brakensiek and Aviad Rubinstein.
\newblock Constant-factor approximation of near-linear edit distance in
  near-linear time.
\newblock In Konstantin Makarychev, Yury Makarychev, Madhur Tulsiani, Gautam
  Kamath, and Julia Chuzhoy, editors, {\em Proccedings of the 52nd Annual {ACM}
  {SIGACT} Symposium on Theory of Computing, {STOC} 2020}, pages 685--698.
  {ACM}, 2020.

\bibitem[BZ16]{belazzougui_zhang}
Djamal Belazzougui and Qin Zhang.
\newblock Edit distance: Sketching, streaming, and document exchange.
\newblock In {\em 2016 IEEE 57th Annual Symposium on Foundations of Computer
  Science (FOCS)}, pages 51--60, 2016.

\bibitem[CDG{\etalchar{+}}18]{CDGKS18}
Diptarka Chakraborty, Debarati Das, Elazar Goldenberg, Michal Kouck{\'{y}}, and
  Michael~E. Saks.
\newblock Approximating edit distance within constant factor in truly
  sub-quadratic time.
\newblock In {\em 59th {IEEE} Annual Symposium on Foundations of Computer
  Science, {FOCS} 2018}, pages 979--990, 2018.

\bibitem[CGK16]{CGK16}
Diptarka Chakraborty, Elazar Goldenberg, and Michal Kouck{\'{y}}.
\newblock Streaming algorithms for embedding and computing edit distance in the
  low distance regime.
\newblock In {\em Proceedings of the 48th Annual {ACM} {SIGACT} Symposium on
  Theory of Computing, {STOC} 2016, Cambridge, MA, USA, June 18-21, 2016},
  pages 712--725, 2016.

\bibitem[CKP19]{rollinghashSODA2019}
Rapha{\"{e}}l Clifford, Tomasz Kociumaka, and Ely Porat.
\newblock The streaming k-mismatch problem.
\newblock In {\em Proceedings of the Thirtieth Annual {ACM-SIAM} Symposium on
  Discrete Algorithms, {SODA} 2019}, pages 1106--1125. {SIAM}, 2019.

\bibitem[CM02]{CM02}
Graham Cormode and S.~Muthukrishnan.
\newblock The string edit distance matching problem with moves.
\newblock In {\em Proceedings of the Thirteenth Annual {ACM-SIAM} Symposium on
  Discrete Algorithms, January 6-8, 2002, San Francisco, CA, {USA.}}, pages
  667--676, 2002.

\bibitem[CV86]{cole1986deterministic}
Richard Cole and Uzi Vishkin.
\newblock Deterministic coin tossing and accelerating cascades: micro and macro
  techniques for designing parallel algorithms.
\newblock In {\em Proceedings of the eighteenth annual ACM symposium on Theory
  of computing (STOC)}, pages 206--219, 1986.

\bibitem[FIM{\etalchar{+}}06]{hamming_sketch3}
Joan Feigenbaum, Yuval Ishai, Tal Malkin, Kobbi Nissim, Martin~J Strauss, and
  Rebecca~N Wright.
\newblock Secure multiparty computation of approximations.
\newblock {\em ACM transactions on Algorithms (TALG)}, 2(3):435--472, 2006.

\bibitem[GKLS22]{ED_compressed_string_Soda22}
Arun Ganesh, Tomasz Kociumaka, Andrea Lincoln, and Barna Saha.
\newblock How compression and approximation affect efficiency in string
  distance measures.
\newblock In {\em Proceedings of the 2022 {ACM-SIAM} Symposium on Discrete
  Algorithms, {SODA}}, pages 2867--2919, 2022.

\bibitem[Gra16]{G16}
Szymon Grabowski.
\newblock New tabulation and sparse dynamic programming based techniques for
  sequence similarity problems.
\newblock {\em Discrete Applied Mathematics}, 212:96--103, 2016.

\bibitem[JNW21]{nelson_edit_sketch}
Ce~Jin, Jelani Nelson, and Kewen Wu.
\newblock An improved sketching algorithm for edit distance.
\newblock In {\em 38th International Symposium on Theoretical Aspects of
  Computer Science, {STACS} 2021,}, volume 187 of {\em LIPIcs}, pages
  45:1--45:16, 2021.

\bibitem[Jow12]{Jow}
Hossein Jowhari.
\newblock Efficient communication protocols for deciding edit distance.
\newblock In {\em Algorithms - {ESA} 2012 - 20th Annual European Symposium,
  Ljubljana, Slovenia, September 10-12, 2012. Proceedings}, pages 648--658,
  2012.

\bibitem[KOR98]{hamming_sketch2}
Eyal Kushilevitz, Rafail Ostrovsky, and Yuval Rabani.
\newblock Efficient search for approximate nearest neighbor in high dimensional
  spaces.
\newblock In {\em Proceedings of the thirtieth annual ACM symposium on Theory
  of computing}, pages 614--623, 1998.

\bibitem[KPS21]{editsketchfocs2021}
Tomasz Kociumaka, Ely Porat, and Tatiana Starikovskaya.
\newblock Small-space and streaming pattern matching with $k$ edits.
\newblock In {\em 2021 IEEE 62nd Annual Symposium on Foundations of Computer
  Science (FOCS)}, pages 885--896, 2021.

\bibitem[KR87]{rabin_karp}
Richard~M. Karp and Michael~O. Rabin.
\newblock Efficient randomized pattern-matching algorithms.
\newblock {\em IBM Journal of Research and Development}, 31(2):249--260, 1987.

\bibitem[KS20]{KS20}
Michal Kouck{\'{y}} and Michael~E. Saks.
\newblock Constant factor approximations to edit distance on far input pairs in
  nearly linear time.
\newblock In Konstantin Makarychev, Yury Makarychev, Madhur Tulsiani, Gautam
  Kamath, and Julia Chuzhoy, editors, {\em Proccedings of the 52nd Annual {ACM}
  {SIGACT} Symposium on Theory of Computing, {STOC} 2020}, pages 699--712.
  {ACM}, 2020.

\bibitem[Lin87]{linial1987distributive}
Nathan Linial.
\newblock Distributive graph algorithms-global solutions from local data.
\newblock In {\em 28th Annual Symposium on Foundations of Computer
  Science,FOCS}, pages 331--335. {IEEE} Computer Society, 1987.

\bibitem[Lin92]{linial1992locality}
Nathan Linial.
\newblock Locality in distributed graph algorithms.
\newblock {\em {SIAM} J. Comput.}, 21(1):193--201, 1992.

\bibitem[LMS98]{LMS98}
Gad~M. Landau, Eugene~W. Myers, and Jeanette~P. Schmidt.
\newblock Incremental string comparison.
\newblock {\em SIAM J. Comput.}, 27(2):557--582, April 1998.

\bibitem[MP80]{MP80}
William~J. Masek and Michael~S. Paterson.
\newblock A faster algorithm computing string edit distances.
\newblock {\em Journal of Computer and System Sciences}, 20(1):18 -- 31, 1980.

\bibitem[OR07]{hamming_to_l1_rabani_2007}
Rafail Ostrovsky and Yuval Rabani.
\newblock Low distortion embeddings for edit distance.
\newblock {\em J. {ACM}}, 54(5):23, 2007.

\bibitem[PL07]{porat2007}
Ely Porat and Ohad Lipsky.
\newblock Improved sketching of hamming distance with error correcting.
\newblock In {\em Combinatorial Pattern Matching, 18th Annual Symposium,
  {CPM}}, volume 4580, pages 173--182. Springer, 2007.

\bibitem[SV94]{lcp-application94}
S{\"{u}}leyman~Cenk Sahinalp and Uzi Vishkin.
\newblock Symmetry breaking for suffix tree construction.
\newblock In {\em Proceedings of the Twenty-Sixth Annual {ACM} Symposium on
  Theory of Computing, 23-25 May 1994, Montr{\'{e}}al, Qu{\'{e}}bec, Canada},
  pages 300--309. {ACM}, 1994.

\bibitem[WF74]{WF74}
Robert~A. Wagner and Michael~J. Fischer.
\newblock The string-to-string correction problem.
\newblock {\em J. ACM}, 21(1):168--173, January 1974.

\end{thebibliography}

% \printbibliography

\end{document}